\documentclass[11pt]{article}

\usepackage[margin=1in]{geometry}
\usepackage{amsthm,amsmath,amsfonts,amssymb}

\usepackage{colonequals}
\usepackage{thmtools,thm-restate}
\usepackage{bbm} 

\include{Qcircuit} 

\usepackage[pagebackref]{hyperref} 
\hypersetup{
    bookmarksnumbered=true, 
    unicode=false, 
    pdfstartview={FitH}, 
    pdftitle={Exponential improvement in precision for simulating sparse Hamiltonians}, 
    pdfauthor={Dominic W. Berry, Andrew M. Childs, Richard Cleve, Robin Kothari, and Rolando D. Somma}, 
    pdfnewwindow=true, 
    colorlinks=true, 
    linkcolor=blue, 
    citecolor=blue, 
    filecolor=blue, 
    urlcolor=blue 
}
\usepackage[all]{hypcap}
\newcommand{\eprint}[1]{\href{http://arxiv.org/abs/#1}{#1}}

\usepackage{tikz}
\usepackage[font=small]{caption}

\renewcommand{\backref}[1]{}

\renewcommand{\backrefalt}[4]{%
\ifcase #1 %
\or 
[p.\ #2]%
\else 
[pp.\ #2]%
\fi}

\usepackage{xspace}
\usepackage[comma,square,numbers,sort]{natbib}
\usepackage{enumitem}

\newtheorem{theorem}{Theorem}[section]
\newtheorem{lemma}[theorem]{Lemma}
\newtheorem{corollary}[theorem]{Corollary}

\theoremstyle{definition}
\newtheorem{definition}{Definition}

\newcommand{\eq}[1]{\hyperref[eq:#1]{(\ref*{eq:#1})}}
\renewcommand{\sec}[1]{\texorpdfstring{\hyperref[sec:#1]{Section~\ref*{sec:#1}}}{Section~\ref*{sec:#1}}}
\newcommand{\thm}[1]{\texorpdfstring{\hyperref[thm:#1]{Theorem~\ref*{thm:#1}}}{Theorem~\ref*{thm:#1}}}
\newcommand{\lem}[1]{\hyperref[lem:#1]{Lemma~\ref*{lem:#1}}}
\newcommand{\prop}[1]{\hyperref[prop:#1]{Proposition~\ref*{prop:#1}}}
\newcommand{\cor}[1]{\hyperref[cor:#1]{Corollary~\ref*{cor:#1}}}
\newcommand{\fig}[1]{\hyperref[fig:#1]{Figure~\ref*{fig:#1}}}
\newcommand{\tab}[1]{\hyperref[tab:#1]{Table~\ref*{tab:#1}}}
\newcommand{\app}[1]{\hyperref[app:#1]{Appendix~\ref*{app:#1}}}

\newcommand{\x}{b}
\newcommand{\h}{\bar h}
\newcommand{\Had}{\mathrm{Had}}
\newcommand{\fudge}{\Upsilon}

\newcommand{\comment}[1]{} 

\newcommand{\R}{{\mathbb{R}}}
\newcommand{\Z}{{\mathbb{Z}}}

\newcommand{\poly}{\mathop{\mathrm{poly}}}
\newcommand{\sgn}{\mathop{\mathrm{sgn}}}

\newcommand{\ceil}[1]{\lceil{#1}\rceil}

\renewcommand{\(}{\left(}
\renewcommand{\)}{\right)}
\newcommand{\defeq}{\colonequals}

\newcommand{\norm}[1]{\|{#1}\|}

\renewcommand{\>}{\rangle}
\newcommand{\<}{\langle}

\newcommand{\id}{\mathbbm{1}}

\newcommand{\be}{\begin{equation}}
\newcommand{\ee}{\end{equation}}

\newcommand{\mj}{\mu}

\DeclareMathOperator{\spn}{span}
\DeclareMathOperator{\wt}{wt}

\newcommand{\parity}{\textsc{parity}}

\newcommand{\approxU}{\tilde U}

\begin{document}


\title{Exponential improvement in precision \\ for simulating sparse Hamiltonians}

\author{
\normalsize Dominic W.\ Berry\thanks{Department of Physics and Astronomy, Macquarie University}
\enskip
\normalsize Andrew M.\ Childs\thanks{Department of Combinatorics \& Optimization and Institute for Quantum Computing, University of Waterloo}\hspace*{1.1mm}$^{,}$\thanks{Canadian Institute for Advanced Research} 
\enskip
\normalsize Richard Cleve\thanks{Cheriton School of Computer Science and Institute for Quantum Computing, University of Waterloo}\hspace*{1.3mm}$^{,\ddag}$
\enskip
\normalsize Robin Kothari$^{\S}$
\enskip
\normalsize Rolando D.\ Somma\thanks{Theory Division, Los Alamos National Laboratory}
}

\date{}

\maketitle

\begin{abstract}
We provide a quantum algorithm for simulating the dynamics of sparse Hamiltonians with complexity sublogarithmic in the inverse error, an exponential improvement over previous methods.  Specifically, we show that a $d$-sparse Hamiltonian $H$ acting on $n$ qubits can be simulated for time $t$ with precision $\epsilon$ using $O\big(\tau \frac{\log(\tau/\epsilon)}{\log\log(\tau/\epsilon)}\big)$ queries and $O\big(\tau \frac{\log^2(\tau/\epsilon)}{\log\log(\tau/\epsilon)}n\big)$ additional 2-qubit gates, where $\tau = d^2 \|{H}\|_{\max} t$. Unlike previous approaches based on product formulas, the query complexity is independent of the number of qubits acted on, and for time-varying Hamiltonians, the gate complexity is logarithmic in the norm of the derivative of the Hamiltonian. Our algorithm is based on a significantly improved simulation of the continuous- and fractional-query models using discrete quantum queries, showing that the former models are not much more powerful than the discrete model even for very small error.  We also simplify the analysis of this conversion, avoiding the need for a complex fault correction procedure.  Our simplification relies on a new form of ``oblivious amplitude amplification'' that can be applied even though the reflection about the input state is unavailable. Finally, we prove new lower bounds showing that our algorithms are optimal as a function of the error.
\end{abstract}

\section{Introduction}
\label{sec:intro}

Simulation of quantum mechanical systems is a major potential application of quantum computers. Indeed, the problem of simulating Hamiltonian dynamics was the original motivation for the idea of quantum computation \cite{Fey82}. Lloyd provided an explicit algorithm for simulating many realistic quantum systems, namely those whose Hamiltonian is a sum of interactions acting nontrivially on a small number of subsystems of limited dimension \cite{Llo96}. If the interactions act on at most $k$ subsystems, such a Hamiltonian is called \emph{$k$-local}. Here we consider the more general problem of simulating sparse Hamiltonians, a natural class of systems for which quantum simulation has been widely studied. Note that $k$-local Hamiltonians are sparse, so algorithms for simulating sparse Hamiltonians can be used to simulate many physical systems.  Sparse Hamiltonian simulation is also useful in quantum algorithms~\cite{AT03,CCJY09,HHL09,CJS13}.

A Hamiltonian is said to be \emph{$d$-sparse} if it has at most $d$ nonzero entries in any row or column. In the sparse Hamiltonian simulation problem,  we are given access to a $d$-sparse Hamiltonian $H$ acting on $n$ qubits via a black box that accepts a row index $i$ and a number $j$ between $1$ and $d$, and returns the position and value of the $j$th nonzero entry of $H$ in row $i$. Given such a black box for $H$, a time $t>0$ (without loss of generality), and an error parameter $\epsilon>0$, our task is to construct a circuit that performs the unitary operation $e^{-iHt}$ with error at most $\epsilon$ using as few queries to $H$ as possible. To develop practical algorithms, we would also like to upper bound the number of additional 2-qubit gates. The \textit{time complexity} of a simulation is the sum of the number of queries and additional 2-qubit gates.

The first efficient algorithm for sparse Hamiltonian simulation was due to Aharonov and Ta-Shma \cite{AT03}.  The key idea (also applied in \cite{CCDFGS03}) is to use edge coloring to decompose the Hamiltonian $H$ into a sum of Hamiltonians $\sum_{j=1}^\eta H_j$, where each $H_j$ is easy to simulate. These terms are then recombined using the Lie product formula, which states that $e^{-iHt} \approx (e^{-iH_1t/r} e^{-iH_2t/r} \cdots e^{-iH_\eta t/r})^r$ for large $r$.  This method gives query complexity $O(\poly(n,d) (\norm{H}t)^2/\epsilon)$, where $\norm{\cdot}$ denotes the spectral norm.  This was later improved using high-order product formulas and more efficient decompositions of the Hamiltonian \cite{Suz91,Chi04,BAC+07,CK11,ChildsWiebe}. The best algorithms of this type \cite{CK11,ChildsWiebe} have query complexity 
\be
d^2(d+\log^*n)\norm{H}t\,\exp\Bigl(O\bigl(\sqrt{\log(d\norm{H}t/\epsilon)}\bigr)\Bigr).
\ee
This complexity is only slightly superlinear in $\norm{H}t$ in that 
$\exp(O(\sqrt{\log(d\norm{H}t/\epsilon)}))$ is asymptotically smaller than 
$(d\norm{H}t/\epsilon)^{\delta}$ for any constant $\delta > 0$; 
however, $\exp(O(\sqrt{\log(d\norm{H}t/\epsilon)}))$ is not polylogarithmic in 
$d\norm{H}t/\epsilon$.

We show the following (where $\norm{H}_{\max}$ denotes the largest entry of $H$ in absolute value).

\begin{restatable}[Sparse Hamiltonian simulation]{theorem}{SPARSE}\label{thm:sparse}
A $d$-sparse Hamiltonian $H$ acting on $n$ qubits can be simulated for time $t$ within error $\epsilon$ with $O\big(\tau \frac{\log(\tau/\epsilon)}{\log\log(\tau/\epsilon)}\big)$ queries and $O\big(\tau \frac{\log^2(\tau/\epsilon)}{\log\log(\tau/\epsilon)}n\big)$ additional 2-qubit gates, where $\tau \defeq d^2 \norm{H}_{\max} t \ge 1$.
\end{restatable}

\noindent
Our algorithm has no query dependence on $n$, improved dependence on $d$ and $t$, and exponentially improved dependence on $1/\epsilon$.
Our new approach to Hamiltonian simulation strictly improves all previous approaches based on product formulas (e.g., \cite{Llo96,AT03,Chi04,BAC+07,CK11}).  An alternative Hamiltonian simulation method based on a quantum walk \cite{Chi10,BC12} is incomparable.  That method has query complexity $O(d \norm{H}_{\max} t/\sqrt{\epsilon})$, so its performance is better in terms of $\norm{H}_{\max} t$ and $d$ but significantly worse in terms of $\epsilon$.  Thus, while suboptimal for (say) constant-precision simulation, the results of \thm{sparse} currently give the best known Hamiltonian simulations as a function of $\epsilon$.

Essentially the same approach used for \thm{sparse} can be applied even when the Hamiltonian is time dependent. The query complexity is unaffected by any such time dependence, except that we take the largest max-norm of the Hamiltonian over all times (i.e., $\tau$ is redefined as $\tau \defeq d^2 h t$ with $h \defeq \max_{s \in [0,t]} \norm{H(s)}_{\max}$).
The number of additional 2-qubit gates is $O\big(\tau \frac{\log(\tau/\epsilon)\log((\tau+\tau')/\epsilon)}{\log\log(\tau/\epsilon)}n\big)$, where $\tau' \defeq d^2 h' t$ with $h' \defeq \max_{s \in [0,t]} \norm{\frac{\mathrm{d}}{\mathrm{d}s}H(s)}$.  This dependence on $h'$ is a dramatic improvement over previous methods for simulating time-dependent Hamiltonians using high-order product formulas \cite{WBHS11}.  Another previous simulation method \cite{PQSV11} also improved the dependence on $h'$, but at the cost of substantially worse dependence on $t$ and $\epsilon$.

While our approach applies to sparse Hamiltonians in general, it can sometimes be improved using additional structure. In particular, consider the case of a $k$-local Hamiltonian acting on a system of qubits. (A $k$-local Hamiltonian acting on subsystems of limited dimension is equivalent to a $k$-local Hamiltonian acting on qubits with an increased value of $k$.) Since a term acting only on $k$ qubits is $2^k$-sparse, we can apply \thm{sparse} with $d=2^k M$, where $M$ is the total number of local terms.  However, by taking the structure of sparse Hamiltonians into account, we find an improved simulation with $\tau$ replaced by $\tilde\tau \defeq 2^k M \norm{H}_{\max} t$.

The performance of our algorithm is optimal or nearly optimal as a function of some of its parameters.  A lower bound of $\Omega(\norm{H}_{\max}t)$ follows from the no-fast-forwarding theorem of \cite{BAC+07}, showing that our algorithm's dependence on $\norm{H}_{\max}t$ is almost optimal.  However, prior to our work, there was no known $\epsilon$-dependent lower bound, not even one ruling out algorithms with no dependence on $\epsilon$.  We show that, surprisingly, our query dependence on $\epsilon$ in \thm{sparse} is optimal.

\begin{restatable}[$\epsilon$-dependent lower bound for Hamiltonian simulation]{theorem}{LOWERBOUND}\label{thm:hamsimlower}
For any $\epsilon>0$, there exists a 2-sparse Hamiltonian $H$ with $\norm{H}_{\max}<1$ such that simulating $H$ with precision $\epsilon$ for constant time requires $\Omega\big(\frac{\log(1/\epsilon)}{\log\log(1/\epsilon)}\big)$ queries.
\end{restatable}

Our Hamiltonian simulation algorithm is based on a connection to the so-called fractional quantum query model.  A result of Cleve, Gottesman, Mosca, Somma, and Yonge-Mallo \cite{CGM+09} shows that this model can be simulated with only small overhead using standard, discrete quantum queries.  While this can be seen as a kind of Hamiltonian simulation, simulating the dynamics of a sparse Hamiltonian appears \emph{a priori} unrelated.  Here we relate these tasks, giving a simple reduction from Hamiltonian simulation to the problem of simulating (a slight generalization of) the fractional-query model, so that improved simulations of the fractional-query model directly yield improvements in Hamiltonian simulation.

To introduce the notion of fractional queries, recall that in the usual model of quantum query complexity, we wish to solve a problem whose input $x \in \{0,1\}^N$ is given by an oracle (or black box) that can be queried to learn the bits of $x$. The measure of complexity, called the query complexity, is the number of times we query the oracle.
More precisely, we are given access to a unitary gate $Q_x$ whose action on the basis states $|j\>|b\>$ for all $j \in [N] \defeq \{1,2,\ldots,N\}$ and $b \in \{0,1\}$ is $ Q_x |j\>|b\> = (-1)^{b x_j}|j\>|b\>$. A quantum query algorithm is a quantum circuit consisting of arbitrary $x$-independent unitaries and $Q_x$ gates. The query complexity of such an algorithm is the total number of $Q_x$ gates used in the circuit.

The query model is often used to study the complexity of evaluating a classical function of $x$.  However, it is also natural to consider more general tasks.  In order of increasing generality, such tasks include state generation \cite{AMRR11}, state conversion \cite{LMRSS11}, and implementing unitary operations \cite{BC12}.  Here we focus on the last of these tasks, where for each possible input $x$ we must perform some unitary operation $U_x$.  Considering this task leads to a strong notion of simulation: to simulate a given algorithm in the sense of unitary implementation, one must reproduce the entire correct output state for every possible input state, rather than simply (say) evaluating some predicate in one bit of the output with a fixed input state.

Since quantum mechanics is fundamentally described by the continuous dynamics of the Schr{\"o}\-ding\-er equation, it is natural to ask if the query model can be made less discrete. In particular, instead of using the gate $Q_x$ for unit cost, what if we can make half a query for half the cost? This perspective is motivated by the idea that if $Q_x$ is performed by a Hamiltonian running for unit time, we can stop the evolution after half the time to obtain half a query. In general we could run this Hamiltonian for time $\alpha \in (0,1]$ at cost $\alpha$. This \emph{fractional-query model} is at least as powerful as the standard (\emph{discrete-query}) model.
More formally, we define the model as follows.

\begin{definition}[Fractional-query model]
For an $n$-bit string $x$, let $Q^\alpha_x$ act as $Q^\alpha_x|j\>|b\> = e^{-i\pi\alpha b x_j}|j\>|b\>$ for all  $j \in [N]$  and  $b \in \{0,1\}$. An algorithm in the fractional-query model is a sequence of unitary gates $U_{m}Q_x^{\alpha_{m}}U_{m-1}\cdots U_{1}Q_x^{\alpha_{1}}U_{0}$, where $U_i$ are arbitrary unitaries and $\alpha_i \in (0,1]$ for all $i$. The fractional-query complexity of this algorithm is $\sum_{i=1}^{m} \alpha_i$ and the total number of fractional-query gates used is $m$. 
\end{definition}

This idea can be taken further by taking the limit as the sizes of the fractional queries approach zero to obtain a continuous variant of the model, called the \emph{continuous-query model} \cite{FG98}.
In this model, we have access to a query Hamiltonian $H_x$ acting as $H_x |j\>|b\> = \pi b x_j|j\>|b\>$. Unlike the fractional- and discrete-query models, this is not a circuit-based model of computation. In this model we are allowed to evolve for time $T$ according to the Hamiltonian given by $H_x + H_D(t)$ for an arbitrary time-dependent driving Hamiltonian $H_D(t)$, at cost $T$.  More precisely, the model is defined as follows.

\begin{definition}[Continuous-query model]
Let $H_x$ act as $H_x |j\>|b\> = \pi b x_j|j\>|b\>$ for all $j \in [N]$  and  $b \in \{0,1\}$. An algorithm in the continuous-query model is specified by an arbitrary $x$-independent driving Hamiltonian $H_D(t)$ for $t \in [0,T]$.  The algorithm implements the unitary operation $U(T)$ obtained by solving the Schr{\"o}dinger equation
\be
\label{eq:schrodinger}
  i \frac{\mathrm{d}}{\mathrm{d}t} U(t) = \big( H_x + H_D(t) \big) U(t)
\ee
with $U(0)=\id$. The continuous-query complexity of this algorithm is the total evolution time, $T$.
\end{definition} 

Because $e^{-i\alpha H_x} = Q^\alpha_x$, running the Hamiltonian $H_x$ with no driving Hamiltonian for time $T=\alpha$ is equivalent to an $\alpha$-fractional query.
In the remainder of this work we omit the subscript $x$ on $Q$ for brevity.

While initial work on the continuous-query model focused on finding analogues of known algorithms \cite{FG98,Moc07}, it has also been studied with the aim of proving lower bounds on the discrete-query model \cite{Moc07}. Furthermore, the model has led to the discovery of new quantum algorithms. In particular, Farhi, Goldstone, and Gutmann \cite{FGG08} discovered an algorithm with continuous-query complexity $O(\sqrt{n})$ for evaluating a balanced binary NAND tree with $n$ leaves, which is optimal. This result was later converted to the discrete-query model with the same query complexity \cite{CCJY09,ACR+10}.

A similar conversion can be performed for any algorithm with a sufficiently well-behaved driving Hamiltonian \cite{Chi10}. However, this leaves open the question of whether continuous-query algorithms can be generically converted to discrete-query algorithms with the same query complexity. This was almost resolved by \cite{CGM+09}, which gave an algorithm that approximates a $T$-query continuous-query algorithm to bounded error with $O\big(T \frac{\log T}{\log\log T}\big)$ discrete queries. This algorithm can be made time efficient \cite{BCG14} (informally, the number of additional 2-qubit gates is close to the query complexity).

However, to approximate a continuous-query algorithm to precision $\epsilon$, the algorithm of \cite{CGM+09} uses $O\big(\frac{1}{\epsilon} \frac{T\log T}{\log\log T}\big)$ queries.  Ideally we would like the dependence on $\epsilon$ to be polylogarithmic, instead of polynomial, in $1/\epsilon$. For example, such behavior would be  desirable when using a fractional-query algorithm as a subroutine. Here we present a significantly improved and simplified simulation of the continuous- and fractional-query models.  In particular, we show the following.

\begin{restatable}[Continuous-query simulation]{theorem}{CQUERYSIM}\label{thm:cquerysim}
An algorithm with continuous- or fractional-query complexity $T \ge 1$ can be simulated with error at most $\epsilon$ with $O\big(T\frac{\log(T/\epsilon)}{\log\log(T/\epsilon)}\big)$ queries.
For continuous-query simulation, if there is a circuit using at most $g$ gates that implements the time evolution due to $H_D(t)$ between any two times $t_1$ and $t_2$ with precision $\epsilon/T$, then the number of additional 2-qubit gates for the simulation is
$O\big(T\frac{\log(T/\epsilon)}{\log\log(T/\epsilon)}[g + \log(\h T/\epsilon)]\big)$, where $\h \defeq \frac{1}{T}\int_{0}^{T}\norm{H_D(t)} \mathrm{d}t$.
\end{restatable}

Since the continuous-query model is at least as powerful as the discrete-query model, a discrete simulation must use $\Omega(T)$ queries, showing our dependence on $T$ is close to optimal. However, as for the problem of Hamiltonian simulation, there was previously no $\epsilon$-dependent lower bound.  Along the lines of \thm{hamsimlower}, we show a lower bound of $\Omega\big(\frac{\log(1/\epsilon)}{\log\log(1/\epsilon)}\big)$ queries for a continuous-query algorithm with $T = O(1)$ (\thm{fqsimlower}), so the dependence of our simulation on $\epsilon$ is optimal.

For the problem of evaluating a classical function of a black-box input, an approach based on an invariant called the \textit{$\gamma_2$ norm} shows that the continuous-query complexity is at most a constant factor smaller than the discrete-query complexity for a bounded-error simulation \cite{LMRSS11}. However, it remains unclear whether the algorithm can be made time efficient and whether the unitary dynamics of a continuous-query algorithm can be simulated (even with bounded error) using $O(T)$ queries. Such a result does hold for state conversion, but its dependence on error is quadratic \cite{LMRSS11}. 
More generally, the optimal tradeoff between $T$ and $\epsilon$ for simulation of continuous-query algorithms using discrete queries---and for simulation of Hamiltonian dynamics---remains open (with or without conditions on the time complexity).

The remainder of this article is organized as follows.  In \sec{overview} we give a high-level overview of the techniques used in our results.  In \sec{main} we describe our simulation of the continuous- and fractional-query models using discrete queries.  In \sec{Ham} we apply these results to Hamiltonian simulation.  In \sec{time} we analyze the time complexity of our algorithms, and in \sec{lb} we prove $\epsilon$-dependent lower bounds showing optimality of their error dependence.  We conclude in \sec{conclusion} with a brief discussion of some open questions.  In \app{proofs}, we provide some proofs of known results for the sake of completeness.

\section{High-level overview of techniques}
\label{sec:overview}

We begin by proving \thm{cquerysim}, our improved simulation of continuous- and fractional-query algorithms. Then we prove \thm{sparse} by reducing an instance of a sparse Hamiltonian simulation problem to an instance of a fractional-query algorithm, which can then be simulated via \thm{cquerysim}. We prove \thm{hamsimlower} using ideas from the no-fast-forwarding theorem from~\cite{BAC+07} and properties of the unbounded-error quantum query complexity of the parity function.

We now sketch the approach for each of the main theorems, highlighting the novel ideas.

\subsection{Continuous-query simulation (\thm{cquerysim})}
\label{sec:thm3}

First consider the simulation of fractional queries using discrete queries.  We show that an algorithm with constant fractional-query complexity can be simulated in the discrete-query model using $O\big(\frac{\log(1/\epsilon)}{\log\log(1/\epsilon)}\big)$ queries (\lem{main}).  The claimed upper bound for simulating a fractional-query algorithm with query complexity $T$ follows
easily by breaking the algorithm into pieces with constant fractional-query complexity.
Since the continuous- and fractional-query models are equivalent (\thm{equiv}), the result for the continuous-query model (\thm{cquerysim}) follows.

We prove \lem{main} in two steps. Let the unitary performed by the constant-query fractional-query algorithm be $V$ and let the (unknown) state it acts on be $\ket{\psi}$. We would like to create the state $V\ket{\psi}$ up to error $\epsilon$. First we construct a circuit $\approxU$ that performs $V$ with amplitude $\sqrt{p}$ up to error $\epsilon$, in the sense that $\approxU$ is within error $\epsilon$ of a unitary $U$ that maps $|0^m\>\ket{\psi}$ to $\sqrt{p}|0^m\>V|\psi\> + \sqrt{1-p} |\Phi^\perp\>$ for some constant $p$ and some state $|\Phi^\perp\>$ with $(|0^m\>\<0^m|\otimes \id)|\Phi^\perp\>=0$.  The existence of such a $\approxU$ that makes $O\big(\frac{\log(1/\epsilon)}{\log\log(1/\epsilon)}\big)$ queries was shown by \cite{CGM+09}. Their strategy is to measure the first $m$ qubits and obtain $V|\psi\>$ with constant probability. If the measurement fails, they recover the original state $|\psi\>$ from $|\Phi^\perp\>$ using a fault-correction procedure, which is itself probabilistic and occasionally fails, requiring a recursive correction algorithm to remove all faults. 
The time-efficient implementation of this recursive fault-correction procedure \cite{BCG14} is cumbersome.

Our alternative approach uses $\approxU$ to deterministically create $V|\psi\>$ without measurements.  We show in general how to create $V|\psi\>$ with a constant number of applications of $\approxU$ when $p$ is a constant.  To do this, we introduce a notion of ``oblivious amplitude amplification'' that can have the same performance as standard amplitude amplification, but that can be applied even when the reflection about the input state is unavailable.  This idea, which is inspired by the in-place QMA amplification procedure of Marriott and Watrous \cite{MW05}, is a general result that can potentially be applied in other contexts.

Most of the algorithm is easily made time efficient, except the preparation of a certain quantum state. However, this state can be prepared efficiently \cite{BCG14} and the result follows.

\subsection{Hamiltonian simulation reduction (\thm{sparse})}
\label{sec:thm1}

Next we describe the main ideas of our Hamiltonian simulation algorithm.
We remove the dependence of the query cost on $n$ with a simple trick involving local edge coloring of bipartite graphs.  This strategy is quite general and can be used to remove $n$-dependence from several known Hamiltonian simulation algorithms.  The improved dependence on $\epsilon$ results from our algorithm for simulating the fractional-query model in the discrete-query model (\thm{cquerysim}).

As mentioned previously, we reduce Hamiltonian simulation to a generalization of the task of simulating the fractional-query model.  Examining the basic Lie product formula $e^{-iHt} \approx (e^{-iH_1t/r} e^{-iH_2t/r} \cdots e^{-iH_\eta t/r})^r$, we see that if $Q_j \defeq e^{-iH_j}$ were query oracles, this would be a fractional-query algorithm using multiple oracles $Q_j$ for time $t$ each. (Note that because the query complexity of the simulation depends only on the total time over which fractional queries are applied rather than the total number of fractional queries, there is no advantage to using higher-order product formulas.) We reduce a fractional-query algorithm that calls each of $\eta$ different query oracles for time $t$ to a fractional-query algorithm that uses query time $\eta t$ with a single query oracle that can perform any $Q_j$.  Thus it suffices to decompose the given Hamiltonian $H$ into a sum of Hamiltonians for which the matrices $Q_j$ can be viewed as query oracles in \thm{cquerysim}. We show such a decomposition (\lem{1sparse}) that yields that stated upper bound.  This algorithm can be made time efficient since it is essentially a reduction to continuous-query simulation.

\subsection{Lower bounds (\thm{hamsimlower} and \thm{fqsimlower})}
\label{sec:thm2}

Finally, we prove lower bounds showing optimality of our algorithms as a function of $\epsilon$ (\thm{hamsimlower} and \thm{fqsimlower}).  The main idea behind both lower bounds is to show a Hamiltonian whose exact simulation for any time $t>0$ allows us to compute the parity of a string with unbounded error, which is as hard as computing parity exactly, requiring $\Omega(n)$ queries \cite{BBC+01,FGGS98}.  Because one must apply the Hamiltonian $\Omega(n)$ times to have nonzero amplitude on a state that encodes the parity, the evolution for constant time only produces the answer at $n$th order in the Taylor series, so the parity is only successfully computed with probability $\Theta(1/n!)$.  To obtain an unbounded-error algorithm for parity, one must simulate this evolution accurately enough to resolve such a small success probability.  Thus we must have $\epsilon = O(1/n!)$, giving the lower bound of $\Omega\big(\frac{\log(1/\epsilon)}{\log\log(1/\epsilon)}\big)$.

\section{From continuous to discrete queries}
\label{sec:main}

In this section we present our improved simulation of continuous or fractional queries in the conventional discrete query model. The main result of this section is \lem{cquerysim}, which establishes the query complexity claimed in \thm{cquerysim}. The time-complexity part of \thm{cquerysim} is established in \sec{time}.

For concreteness, we quantify the distance between unitaries $U$ and $V$ with the function $\norm{U-V}$ and the distance between states $|\psi\>$ and $|\phi\>$ with the function $\norm{|\psi\>-|\phi\>}$.  As the error ultimately appears inside a logarithm, the precise choice of distance measure is not significant.

We begin by recalling the equivalence of the continuous- and fractional-query models for any error $\epsilon>0$.  An explicit simulation of the continuous-query model by the fractional-query model was provided by \cite{CGM+09}; the proof is a straightforward application of a result of \cite{HR90}.  The other direction is apparently folklore (e.g., both directions are implicitly assumed in \cite{Moc07}); we provide a short proof in \app{equiv} for completeness.

\begin{restatable}[Equivalence of continuous- and fractional-query models]{theorem}{EQUIV}
\label{thm:equiv}
For any $\epsilon>0$, any algorithm with continuous-query complexity $T$ can be implemented with fractional-query complexity $T$ with error at most $\epsilon$ and $m = O(\h T^2/\epsilon)$ fractional-query gates, where $\h \defeq \frac{1}{T} \int_0^T \norm{H_D(t)} \, \mathrm{d}t$ is the average norm of the driving Hamiltonian. Conversely, any algorithm with fractional-query complexity $T$ can be implemented with continuous-query complexity $T$ with error at most $\epsilon$. 
\end{restatable}

Since the two models are equivalent, it suffices to convert a fractional-query algorithm to a discrete-query algorithm. We start with a fractional-query algorithm that makes at most 1 query.  The result for multiple queries (\lem{cquerysim}) follows straightforwardly.

\begin{restatable}{lemma}{MAIN}
\label{lem:main}
Any algorithm in the fractional-query model with query complexity at most 1 can be implemented with $O\big(\frac{\log(1/\epsilon)}{\log\log(1/\epsilon)}\big)$ queries in the discrete-query model with error at most $\epsilon$.
\end{restatable}

The construction of the algorithm in this main lemma can be viewed in two steps. First, we show how to unitarily construct a superposition of the required state along with a label in state $\ket{0^{m+1}}$ and another state whose label is orthogonal. The construction is similar to that in \cite{CGM+09,BCG14}; the main difference is that we do not measure the state of the label. (This step is shown in the sequence \lem{gadget}, \lem{segment}, and \lem{approxsegment}.) Then, in the second step, rather than performing a fault-correction procedure upon seeing a measurement outcome other than $0^{m+1}$, we perform the underlying unitary operation in the first step three times (one of which is backwards) in conjunction with certain reflections to arrive at the required state. This step can be viewed as applying a generalization of amplitude amplification that is shown in \lem{oaa}.

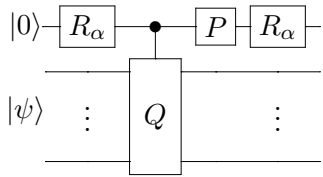
\begin{figure}[ht]
\[ 
\Qcircuit @R=0.4em @C=0.5em {
 \push{\ket{0}}&\gate{R_\alpha}&\ctrl{1}&\gate{P}&\gate{R_\alpha}&\qw\\
 \push{\hphantom{|\psi\>}}& \qw &\multigate{3}{Q}&\qw &\qw &\qw\\
 \push{\ket{\psi}}&  \vdots  & & & \vdots & \\
 &&&&\\
 \push{\hphantom{|\psi\>}}& \qw &\ghost{Q}&\qw &\qw &\qw
}
\]
\caption{\label{fig:gadget}The fractional-query gadget.
After performing the controlled-$Q$ operation on the target state $|\psi\>$, the operation $Q^\alpha$ is performed with amplitude depending on $\alpha$.}
\end{figure}

The first step of the construction uses the fractional-query gadget \cite[Section II.B]{CGM+09} shown in \fig{gadget}.  This gadget behaves as follows, as we show in \app{approxsegment}.

\begin{restatable}[Gadget Lemma \cite{CGM+09}]{lemma}{GADGET}
\label{lem:gadget}
Let $Q$ be a unitary matrix with eigenvalues $\pm 1$; let $\alpha \in [0,1]$. The circuit in \fig{gadget}, with $R_\alpha \defeq \frac{1}{\sqrt{c+s}}\left(\begin{smallmatrix}\sqrt{c} & \sqrt{s} \\ \sqrt{s}  & -\sqrt{c}\end{smallmatrix}\right)$ and $P \defeq \left(\begin{smallmatrix}1 & 0 \\ 0  & i\end{smallmatrix}\right)$, performs the map
\be
|0\>|\psi\> \mapsto \sqrt{q_\alpha}|0\>e^{-i\pi\alpha/2}Q^\alpha |\psi\> + \sqrt{1-q_\alpha}|1\>|\phi\>
\ee
for some state $|\phi\>$, where $c \defeq \cos(\pi \alpha/2)$, $s \defeq \sin(\pi \alpha/2)$, $q_\alpha \defeq 1/(c+s)^2 = 1/(1+\sin(\pi\alpha))$, and $Q^\alpha = \frac{1}{2}(\id+Q) + e^{-i\pi\alpha} \frac{1}{2}(\id-Q) = e^{-i \pi \alpha/2} (c \id + i s Q)$.
\end{restatable}

While the proof in \app{approxsegment} shows that $|\phi\> = e^{-i \pi/4} Q^{-1/2}|\psi\>$, we do not use this fact in our analysis, in contrast to previous approaches \cite{CGM+09,BCG14}.

Note that while we have defined the fractional-query model to use fractions $\alpha \in (0,1]$, a similar simulation could be applied if we allowed negative fractional-time evolutions with $\alpha \in [-1,1]$.  In particular, we could define 
$s = \sin(\pi|\alpha|/2)$, $P = (\begin{smallmatrix}1 & 0 \\ 0 & i \sgn(\alpha)\end{smallmatrix})$
and carry through an analogous analysis.  However, for simplicity, we restrict our attention to the model with only positive fractional time evolutions.

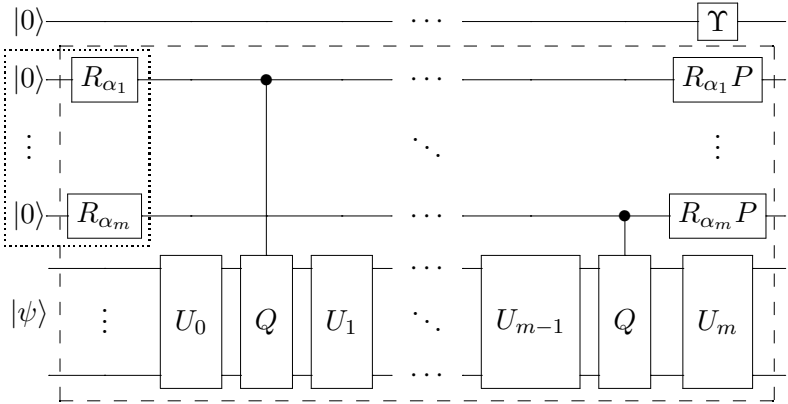
\begin{figure}[ht]
\[ 
\Qcircuit @R=.6em @C=0.65em {
\push{|0\>} \gategroup{2}{2}{9}{11}{.8em}{--}  \gategroup{2}{1}{5}{2}{.5em}{.}  & \qw  & \qw & \qw & \qw & \qw &\push{\cdots}& & \qw & \qw&  \gate{\fudge} & \qw  \\
\push{|0\>} & \gate{R_{\alpha_1}} & \qw & \ctrl{4} & \qw & \qw &{\cdots}& & \qw & \qw& \gate{R_{\alpha_1}P}  & \qw  \\
\push{\vdots} &        &    &   & &   &\ddots& &  &    & \vdots &   \\
   &        &        &   & &   & & &  &    &  &   \\
\push{|0\>} & \gate{R_{\alpha_{m}}} & \qw & \qw & \qw & \qw &{\cdots}& & \qw & \ctrl{1} & \gate{R_{\alpha_{m}}P} & \qw \\
\push{\hphantom{|\psi\>}}& \qw & \multigate{3}{U_0} & \multigate{3}{Q} & \multigate{3}{U_1} & \qw &{\cdots}& & \multigate{3}{U_{m-1}} & \multigate{3}{Q}& \multigate{3}{U_{m}} &  \qw  \\
\push{|\psi\>}& {\vdots} & &  &  &  &\ddots& &  &  &  &    \\
& & &  &  &  && &  &  &  &   \\
\push{\hphantom{|\psi\>}} & \qw & \ghost{U_0} & \ghost{Q} & \ghost{U_1} & \qw &{\cdots}& & \ghost{U_{m-1}} & \ghost{Q}& \ghost{U_{m}} &    \qw  \\
}
\]
\caption{\label{fig:segment}A segment to implement the fractional-query algorithm.
The segment consists of many concatneated applications of the fractional-query gadget, interspersed with $x$-independent unitaries $U_i$.
The state preparation is indicated in the dotted box, and the main operation is performed by the circuit in the dashed box.
The additional ancilla at the top is introduced to reduce the amplitude for performing the correct operation to exactly $1/2$.}
\end{figure}

We now collect the gadgets into segments as shown in \fig{segment} and show that, with an appropriate choice of parameters, a segment implements a fractional-query algorithm with constant query complexity with amplitude $1/2$. This specific choice facilitates one-step exact oblivious amplitude amplification. Other than this choice of constant, this lemma is the same as in \cite{CGM+09}. For completeness, we provide a proof in \app{approxsegment}.

\begin{restatable}[Segment Lemma]{lemma}{SEGMENT}
\label{lem:segment}
Let $V$ be a unitary implementable by a fractional-query algorithm with query complexity at most $1/5$, i.e., there exists an $m$ such that $V=U_{m}Q^{\alpha_{m}}U_{m-1}\cdots U_{1}Q^{\alpha_{1}}U_{0}$ with $\alpha_i \geq 0$ for all $i$ and $\sum_{i=1}^{m} \alpha_i \leq 1/5$. Let $P$ and $R_\alpha$ be as in \lem{gadget}. Then there exists a unitary $\fudge$ on the additional ancilla such that the circuit in \fig{segment} performs the map
\be
|0^{m+1}\>|\psi\> \mapsto \frac{1}{2}|0^{m+1}\>e^{i\vartheta}V|\psi\> + \frac{\sqrt{3}}{2}|\Phi^\perp\>
\ee
for some state $|\Phi^\perp\>$ satisfying $(|0^{m+1}\>\<0^{m+1}| \otimes \id)|\Phi^\perp\>=0$ and some $\vartheta \in [0,2\pi)$.
\end{restatable}

Although the segment in \fig{segment} makes $m$ queries, it is possible to approximate this segment within precision $\epsilon$ using only $O(\frac{\log(1/\epsilon)}{\log\log(1/\epsilon)})$ queries. To get some intuition for why this is possible, note that the state on the control registers decides how many queries are performed. For example, if all the control registers were set to $|0\>$ when the controlled-$Q$ gates act, then no queries would be performed, even though the circuit contains $m$ query gates. In general, the number of queries performed when the control registers are set to $|\x_1,\x_2,\ldots,\x_m\>$ is the Hamming weight of $\x$.  In \fig{segment}, the state of the control registers has very little overlap with high-weight states, so we can approximate that state with one that has no overlap with high-weight states.  We then show how to rearrange such a circuit to obtain a new circuit that uses very few query gates.

This lemma follows the same proof structure as Section II.C of \cite{CGM+09}, but is more general since we do not restrict all the fractional queries to have the same value of $\alpha$. This change requires us to use a version of the Chernoff bound for independent (but not necessarily identically distributed) random variables instead of the one used in \cite{CGM+09}. The lemma is proved in \app{approxsegment}.

\begin{restatable}[Approximate Segment Lemma]{lemma}{APPROXSEGMENT}
\label{lem:approxsegment}
Let $V$ be a unitary implementable by a fractional-query algorithm with query complexity at most $1/5$. Then for any $\epsilon>0$, there exists a unitary quantum circuit that makes $O\big(\frac{\log(1/\epsilon)}{\log\log(1/\epsilon)}\big)$ discrete queries and, within error $\epsilon$, performs a unitary $U$ acting as
\be
 U |0^{m+1}\>|\psi\> = 
  \frac{1}{2}|0^{m+1}\>e^{i\vartheta}V|\psi\> + \frac{\sqrt{3}}{2}|\Phi^\perp\>
\ee
for some state $|\Phi^\perp\>$ satisfying $(|0^{m+1}\>\<0^{m+1}| \otimes \id)|\Phi^\perp\>=0$ and some $\vartheta \in [0,2\pi)$.
\end{restatable}

Up to this point our proof is similar to previous approaches \cite{CGM+09,BCG14}.  In those previous approaches, the map of \lem{approxsegment} was used to probabilistically create the desired state by measuring the first $m+1$ qubits. With constant probability we obtain the desired state, but in the other case we have a fault and have to recover the original input state. This recovery stage required a fault-correction procedure that is difficult to analyze and considerably harder to make time efficient. 

We avoid these difficulties by introducing oblivious amplitude amplification.  Given a unitary $U$ that implements another unitary $V$ with some amplitude (in a certain precise sense), this idea allows one to use a version of amplitude amplification to give a better implementation of $V$.  In particular, as in amplitude amplification, if the amplitude for implementing $V$ is known, we can exactly perform $V$.

In standard amplitude amplification, to amplify the ``good'' part of a state, we need to be able to reflect about the state itself and the projector onto the good subspace. While the latter is easy in our application, we cannot reflect about the unknown input state.  Nevertheless, we show the following.

\begin{lemma}[Oblivious amplitude amplification]
\label{lem:oaa}
Let $U$ and $V$ be unitary matrices on $\mj+n$ qubits and $n$ qubits, respectively, and let $\theta \in (0,\pi/2)$.  Suppose that for any $n$-qubit state $|\psi\>$,
\be
  U|0^\mj\>|\psi\> = \sin(\theta) |0^\mj\>V|\psi\> + \cos(\theta) |\Phi^\perp\>,
\ee
where $|\Phi^\perp\>$ is an $(\mj+n)$-qubit state that depends on $|\psi\>$ and satisfies $\Pi|\Phi^\perp\>=0$, where $\Pi \defeq |0^{\mj}\>\<0^{\mj}| \otimes \id$.  Let $R \defeq 2\Pi-\id$ and $S \defeq -U R U^\dag R$.  Then for any $\ell \in \Z$,
\begin{align}
  S^\ell U |0^\mj\>|\psi\> &= \sin\bigl((2\ell+1)\theta\bigr) |0^\mj\>V|\psi\> + \cos\bigl((2\ell+1)\theta\bigr) |\Phi^\perp\>.
\end{align}
\end{lemma}

Note that $R$ is not the reflection about the initial state, so \lem{oaa} does not follow from amplitude amplification alone.  However, in the context described in the lemma, it suffices to use a different reflection.

The motivation for oblivious amplitude amplification comes from work of Marriott and Watrous on in-place amplification of QMA \cite{MW05} (see also related work on quantum rewinding for zero-knowledge proofs \cite{Wat09} and on using amplitude amplification to obtain a quadratic improvement \cite{NWZ09}).  Specifically, the following technical lemma shows that amplitude amplification remains in a certain 2-dimensional subspace in which it is possible to perform the appropriate reflections.

\begin{lemma}[2D Subspace Lemma]
\label{lem:2d}
Let $U$ and $V$ be unitary matrices on $\mj+n$ qubits and $n$ qubits, respectively, and let $p \in (0,1)$.  Suppose that for any $n$-qubit state $|\psi\>$,
\be
  U|0^\mj\>|\psi\> = \sqrt{p}|0^\mj\>V|\psi\> + \sqrt{1-p}|\Phi^{\perp}\>,
\ee
where $|\Phi^{\perp}\>$ is an $(\mj+n)$-qubit state that depends on $|\psi\>$ and satisfies $\Pi|\Phi^\perp\>=0$, where $\Pi \defeq |0^{\mj}\>\<0^{\mj}| \otimes \id$. Then the state $|\Psi^\perp\>$ defined by the equation
\be 
U |\Psi^\perp\>\defeq\sqrt{1-p}|0^\mj\>V|\psi\> - \sqrt{p}|\Phi^\perp\>
\ee 
is orthogonal to $|\Psi\> \defeq |0^\mj\>|\psi\>$ and satisfies $\Pi|\Psi^\perp\>=0$.
\end{lemma}

\begin{proof}
For any $|\psi\>$, let $|\Phi\> \defeq |0^\mj\>V|\psi\>$. Then for all $|\psi\>$, we have 
\begin{align}
U|\Psi\> &= \sqrt{p}|\Phi\> + \sqrt{1-p}|\Phi^{\perp}\> \label{eq:psi} \\
U|\Psi^{\perp}\> &= \sqrt{1-p}|\Phi\> - \sqrt{p}|\Phi^{\perp}\>, \label{eq:psiperp}
\end{align}
where $\Pi|\Phi^{\perp}\> = 0$. By taking the inner product of these two equations, we get $\<\Psi|\Psi^{\perp}\> = 0$. The lemma asserts that not only is $|\Psi^{\perp}\>$ orthogonal to $|\Psi\>$, but also $\Pi |\Psi^{\perp}\> = 0$.

To show this, consider the operator 
\be
  Q \defeq (\<0^\mj| \otimes \id)U^{\dag}\Pi U(|0^\mj\> \otimes \id).
\ee
For any state $|\psi\>$,
\be
  \<\psi|Q|\psi\> 
  = \norm{\Pi U |0^\mj\>|\psi\>}^2 
  = \norm{\Pi (\sqrt{p}|\Phi\> + \sqrt{1-p}|\Phi^{\perp}\>)}^2 
  = \norm{\sqrt{p}|\Phi\>}^2 
  = p.
\ee
In particular, this holds for a basis of eigenvectors of $Q$, so $Q = p \id$.

Thus for any $|\psi\>$, we have
\be
  p|\psi\> = Q|\psi\> = (\<0^\mj| \otimes \id)U^{\dag}\Pi U(|0^\mj\> \otimes \id)|\psi\> = (\<0^\mj| \otimes \id)U^{\dag}\Pi U|\Psi\> = \sqrt{p} (\<0^\mj| \otimes \id)U^{\dag}|\Phi\>.
\ee
From \eq{psi} and \eq{psiperp} we get $U^{\dag} |\Phi\> = \sqrt{p} |\Psi\> +\sqrt{1-p} |\Psi^{\perp}\>$. Plugging this into the previous equation, we get
\begin{align}
  p |\psi\> 
  &= \sqrt{p}(\<0^\mj|\otimes \id) ( \sqrt{p} |\Psi\> + \sqrt{1-p} |\Psi^{\perp}\>)
  = p |\psi\> + \sqrt{p(1-p)} (\<0^\mj|\otimes \id)|\Psi^{\perp}\>.
\end{align}
This gives us $\sqrt{p(1-p)}(\<0^\mj|\otimes \id)|\Psi^{\perp}\>=0$. Since $p \in (0,1)$, this implies $\Pi |\Psi^{\perp}\> = 0$.
\end{proof}

Note that this fact can also be viewed as a consequence of Jordan's Lemma \cite{Jor75}, which decomposes the space into a direct sum of 1- and 2-dimensional subspaces that are invariant under the projectors $\Pi$ and $U^\dagger \Pi U$.  In this decomposition, $\Pi$ and $U^\dagger \Pi U$ are rank-1 projectors within each 2-dimensional subspace. Let $|0\>|\psi_i\>$ denote the eigenvalue-1 eigenvector of $\Pi$ within the $i$th 2-dimensional subspace $S_i$.  Since $S_i$ is invariant under $U^\dagger \Pi U$, the state $U^\dagger \Pi U|0\>|\psi_i\> = \sqrt{p}U^\dagger|0\>V|\psi_i\>$ belongs to $S_i$. Let $|\Phi_i^{\perp}\>$ be such that $|0\>|\psi_i\> = U^{\dag}(\sqrt{p}|0\>V|\psi_i\> + \sqrt{1-p}|\Phi_i^{\perp}\>)$. Then $|\Psi_i^{\perp}\> \defeq U^{\dag} (\sqrt{1-p}|0\>V|\psi_i\> - \sqrt{p}|\Phi_i^{\perp}\>)$ is in $S_i$, since it is a linear combination of $|0\>|\psi_i\>$ and $U^\dagger \Pi U|0\>|\psi_i\>$. However, $|\Psi_i^{\perp}\>$ is orthogonal to $|0\>|\psi_i\>$ and is therefore an eigenvalue-0 eigenvector of $\Pi$, since $\Pi$ is a rank-1 projector in $S_i$. Thus for each $i$, $|\psi_i\>$ and $|\Psi_i^\perp\>$ satisfy the conditions of the lemma.  We claim that the number of 2-dimensional subspaces (and hence the number of states $|\psi_i\>$) is $2^n$.  There are at most $2^n$ such subspaces since $\Pi$ has rank $2^n$ and is rank-1 in each subspace. There also must be at least $2^n$ 2-dimensional subspaces, since otherwise there would be a state $|0\>|\psi\>$ that is in a 1-dimensional subspace, i.e., is invariant under both $\Pi$ and $U^\dagger \Pi U$. This is not possible because $U^\dagger \Pi U$ acting on $|0\>|\psi\>$ yields  $\sqrt{p}U^\dagger|0\>V|\psi\>$, which is a subnormalized state since $p<1$.  Finally, since there are $2^n$ linearly independent $|\psi_i\>$, an arbitrary state $|\psi\>$ can be written as a linear combination of $|\psi_i\>$, and the result follows.

With the help of \lem{2d} we can prove \lem{oaa}.

\begin{proof}[Proof of \protect{\lem{oaa}}]
Since \lem{2d} shows that the evolution occurs within a two-dimensional subspace (or its image under $U$), the remaining analysis is essentially the same as in standard amplitude amplification.  For any $|\psi\>$, we define  $|\Psi\>\defeq |0^\mj\>|\psi\>$ and $|\Phi\>\defeq |0^\mj\>V|\psi\>$, so that
\begin{align}
	U |\Psi\> &= \sin(\theta) |\Phi\> + \cos(\theta) |\Phi^\perp\>,
\end{align}
where $\theta\in(0,\pi/2)$ is such that $\sqrt{p}=\sin(\theta)$.
We also define $|\Psi^\perp\>$ through the equation
\begin{align}
	U |\Psi^\perp\> \defeq \cos(\theta) |\Phi\> - \sin(\theta) |\Phi^\perp\>.
\end{align}
By \lem{2d}, we know that $\Pi|\Psi^\perp\>=0$. Using these two equations, we have 
\begin{align}
	U^\dag |\Phi\> 
	&= \sin(\theta) |\Psi\> + \cos(\theta) |\Psi^\perp\> \\
	U^\dag |\Phi^\perp\>
	&= \cos(\theta) |\Psi\> - \sin(\theta) |\Psi^\perp\>.
\end{align}
Then a straightforward calculation gives
\begin{align}
  S |\Phi\> 
  &= -U R U^\dag |\Phi\> \nonumber \\
  &= -U R (\sin(\theta) |\Psi\> + \cos(\theta) |\Psi^\perp\>) \nonumber \\
  &= -U (\sin(\theta) |\Psi\> - \cos(\theta) |\Psi^\perp\>) \nonumber \\
  &= \bigl(\cos^2(\theta)-\sin^2(\theta)\bigr) |\Phi\> 
     - 2\cos(\theta)\sin(\theta) |\Phi^\perp\> \nonumber \\
  &= \cos(2\theta) |\Phi\> - \sin(2\theta) |\Phi^\perp\>.
\end{align}
Similarly,
\begin{align}
  S |\Phi^\perp\>
  &= U R U^\dag |\Phi^\perp\> \nonumber \\
  &= U R (\cos(\theta) |\Psi\> - \sin(\theta) |\Psi^\perp\>) \nonumber \\
  &= U (\cos(\theta) |\Psi\> + \sin(\theta) |\Psi^\perp\>) \nonumber \\
  &= 2\cos(\theta)\sin(\theta) |\Phi\> 
     + \bigl(\cos^2(\theta)-\sin^2(\theta)\bigr) |\Phi^\perp\> \nonumber \\
  &= \sin(2\theta) |\Phi\> + \cos(2\theta) |\Phi^\perp\>.
\end{align}
Thus we see that $S$ acts as a rotation by $2\theta$ in the subspace $\spn\{|\Phi\>,|\Phi^\perp\>\}$, and the result follows. 
\end{proof}

We are now ready to complete the proof of \lem{main} using \lem{approxsegment} and \lem{oaa}.

\begin{proof}[Proof of \protect\lem{main}]
We are given a fractional-query algorithm that makes at most 1 query. This can be split into 5 steps that make at most 1/5 queries each in the fractional-query model.
We perform the analysis for these steps of size $1/5$; the difference is only a constant factor that does not affect the asymptotics.
We convert this fractional-query algorithm into a discrete-query algorithm with some error.

From \lem{approxsegment}, we know that for any such fractional-query algorithm $V$, there is an algorithm that makes $O\big(\frac{\log(1/\epsilon)}{\log\log(1/\epsilon)}\big)$ discrete queries and maps the state $|0^{m+1}\>|\psi\>$ to a state that is at most $\epsilon$ far from $\frac{1}{2}|0^{m+1}\>e^{i\vartheta}V|\psi\> + \frac{\sqrt{3}}{2}|\Phi\>$, for some state $|\Phi\>$ that satisfies $(|0^{m+1}\>\<0^{m+1}| \otimes \id)|\Phi\>=0$ and some $\vartheta \in [0,2\pi)$.
We wish to perform the unitary $V$ on the input state $|\psi\>$ approximately.

The unitary operation $U$ defined in \lem{approxsegment} maps $|0^{m+1}\>|\psi\> \mapsto \frac{1}{2}|0^{m+1}\>e^{i\vartheta}V|\psi\> + \frac{\sqrt{3}}{2}|\Phi\>$.  The operation $U$ satisfies the conditions of \lem{oaa} with $\mj=m+1$ and $\sin^2(\theta)=1/4$. 
Thus a single application of $S$ (using three applications of $U$) would produce the state $V|\psi\>$ exactly.

While we cannot necessarily perform $U$, using \lem{approxsegment} we can perform another unitary operation $\approxU$ that is within error $\epsilon/3$ of $U$. Since we only perform the unitary three times, we obtain a state $\epsilon$-close to $V|\psi\>$ when we use $\approxU$ instead of $U$.
\end{proof}

By straightforwardly concatenating such simulations with sufficiently small error, we obtain simulations for longer times. This establishes the following lemma, which is the query-complexity part of \thm{cquerysim}.

\begin{lemma}
\label{lem:cquerysim}
An algorithm with continuous- or fractional-query complexity $T \ge 1$ can be simulated with error at most $\epsilon$ with $O\big(T\frac{\log(T/\epsilon)}{\log\log(T/\epsilon)}\big)$ queries.
\end{lemma}

\begin{proof}
Given an algorithm that runs for time $T$ in the continuous-query model, we can convert it to an algorithm with fractional-query complexity $T$ with error at most $\epsilon/2$ using \thm{equiv}. Given a fractional-query algorithm that makes $T$ queries, we can divide it into $\ceil{T}$ pieces that make at most 1 query each and invoke \lem{main} with error $\epsilon/2\ceil{T}$ to obtain $\ceil{T}$ discrete-query algorithms, each of which makes $O\big(\frac{\log(\ceil{T}/\epsilon)}{\log\log(\ceil{T}/\epsilon)}\big)$ queries.  When run sequentially on the input state, they yield an output that is $\epsilon/2$-close to the correct output (by subadditivity of error). Thus the final state has error at most $\epsilon$.
\end{proof}

\section{Hamiltonian simulation}
\label{sec:Ham}

We now apply the results of the previous section to give improved algorithms for simulating sparse Hamiltonians.   The main result of this section is the reduction from an instance of the sparse Hamiltonian simulation problem to a fractional-query algorithm, which establishes \lem{sparse}, the query-complexity part of \thm{sparse}. The time-complexity part of \thm{sparse} is established in \sec{time}.

To see the connection between the fractional-query model and Hamiltonian simulation, consider the example of a Hamiltonian $H = H_1 + H_2$, where $H_1$ and $H_2$ have eigenvalues $0$ and $\pi$, so that $e^{-iH_1}$ and $e^{-iH_2}$  have eigenvalues $\pm 1$.  From the Lie product formula, we have $e^{-i(H_1+H_2)T} \approx (e^{-iH_1T/r}e^{-iH_2T/r})^r$ for large $r$.  If we think of $H_1$ and $H_2$ as query Hamiltonians, this is a fractional-query algorithm that makes $T$ queries to each Hamiltonian. 
We might therefore expect that $O\big(T\frac{\log(T/\epsilon)}{\log\log(T/\epsilon)}\big)$ discrete queries to $e^{-iH_1}$ and $e^{-iH_2}$ suffice to implement $e^{-i(H_1+H_2)T}$ to precision $\epsilon$. Here we do this by generalizing the results of the previous section to allow multiple fractional-query oracles.

For a set $\mathcal{Q} = \{Q_1, \ldots, Q_\eta\}$ of unitary matrices with eigenvalues $\pm 1$, we say $U$ is a fractional-query algorithm over $\mathcal{Q}$ with cost $T$ if $U$ can be written as $U_\lambda Q^{\alpha_\lambda}_{i_\lambda}U_{\lambda-1}\cdots U_{1}Q^{\alpha_1}_{i_1}U_0$, where $0<\alpha_i\leq 1$, $\sum_{i=1}^{\lambda} \alpha_i = T$, and $i_j \in [\eta]$ for all $j\in [\lambda]$.

\begin{restatable}[Multiple-query model]{theorem}{MULTIQUERY}
\label{thm:multiquery}
Let $\mathcal{Q} = \{Q_1, \ldots, Q_\eta\}$ be a set of unitaries with eigenvalues $\pm 1$. Let $U$ be a fractional-query algorithm over $\mathcal{Q}$ with cost $T$.  Let $Q \defeq \sum_{j=1}^\eta |j\>\<j| \otimes Q_j$. Then $U$ can be implemented by a circuit that makes $O\big(T\frac{\log(T/\epsilon)}{\log\log(T/\epsilon)}\big)$ queries to $Q$ with error at most $\epsilon$.
\end{restatable}

\begin{proof}
We prove this by reduction to \thm{cquerysim}. We know that $U$ can be written in the form $U = U_\lambda Q^{\alpha_\lambda}_{i_\lambda}U_{\lambda-1}\cdots U_{1}Q^{\alpha_1}_{i_1}U_0$, where $0<\alpha_i\leq 1$, $\sum_{i=1}^{\lambda} \alpha_i = T$, and $i_j \in [\eta]$ for all $j\in [\lambda]$.  

We first express $U$ as a fractional-query algorithm over $\mathcal{Q}$ with cost $T$. To do this, we add an extra control register to the original circuit for $U$.  This register holds the index $i_j$ of the next query to be performed.  We start with this register initialized to $|0\>$.  Let $V_0$ be any unitary that maps $|0\>$ to $|i_1\>$.  The action of $Q^{\alpha_1}_{i_1}U_0$ on any state $|\psi\>$ is the same as the action of $Q^{\alpha_1}(V_0 \otimes U_0)$ on the second register of $|0\>|\psi\>$. Similarly, for all $j \in [\lambda]$, let $V_j$ be any unitary that maps $|i_j\>$ to $|i_{j+1}\>$, where $i_{\lambda+1} \defeq 0$.  Thus  the circuit $(V_\lambda \otimes U_\lambda)Q^{\alpha_\lambda}(V_{\lambda-1} \otimes U_{\lambda-1})\cdots (V_1 \otimes U_{1})Q^{\alpha_1}(V_0 \otimes U_0)$ maps $|0\>|\psi\>$ to $|0\>U|\psi\>$.

This construction gives a fractional-query algorithm with fractional-query complexity $T$ given oracle access to $Q$. Since $Q$ has eigenvalues $\pm 1$, we can invoke \thm{cquerysim} to give a discrete-query algorithm that makes $O\big(T\frac{\log(T/\epsilon)}{\log\log(T/\epsilon)}\big)$ queries to $Q$ and performs $U$ up to error $\epsilon$.
\thm{cquerysim} assumes the queries are diagonal in the computational basis, whereas here we assume only that $Q$ has eigenvalues $\pm1$.
However, these two scenarios are equivalent since the target system can be considered in a basis where $Q$ is diagonal.
Therefore \thm{cquerysim} applies to the slightly more general scenario considered here.
\end{proof}

This theorem allows us to simulate a Hamiltonian $H = H_1 + \cdots + H_\eta$ for time $t$ using resources that scale only slightly superlinearly in $\eta t$, provided each $H_j$ has eigenvalues $0$ and $\pi$ (or more generally, by rescaling, provided each $H_j$ has the same two eigenvalues).  For any $\epsilon > 0$, there is a sufficiently large $r$ so that $e^{-iHt}$ is $\epsilon$-close to $(e^{-iH_1t/r} \cdots e^{-iH_\eta t/r})^r$, which is of the form required by \thm{multiquery} if $e^{-iH_j}$ has eigenvalues $\pm 1$.  Since $\norm{e^{-iHt} - (e^{-iH_1t/r} \cdots e^{-iH_\eta t/r})^r} = O((\eta\bar{h}t)^2/r)$, where $\bar{h} \defeq \max_j \norm{H_j}$ \cite{BAC+07}, 
choosing $r = \Omega((\eta\bar{h}t)^2/\epsilon)$ is sufficient to achieve an $\epsilon$-approximation.
Since our Hamiltonians $H_j$ have constant norm, we have $\bar{h} = O(1)$ and get the following corollary.

\begin{corollary}
\label{cor:hamsim}
For a Hamiltonian $H = \sum_{j=1}^\eta H_j$, where $H_j$ has eigenvalues $0$ and $\pi$ for all $j \in [\eta]$, define $Q \defeq  \sum_j |j\>\<j| \otimes e^{-iH_j}$. The unitary $e^{-iHt}$ can be implemented by a fractional-query algorithm over $Q$, up to error $\epsilon$, with query complexity $\tau = \eta t$ and $O(\eta^3 t^2/\epsilon)$ fractional-query gates. Thus $e^{-iHt}$ can be implemented up to error $\epsilon$ by a circuit with $O\big(\tau\frac{\log(\tau/\epsilon)}{\log\log(\tau/\epsilon)}\big)$ invocations of $Q$.
\end{corollary}

To simulate arbitrary sparse Hamiltonians, we decompose them into Hamiltonians with this property.  To do this we first decompose the Hamiltonian into a sum of 1-sparse Hamiltonians (with at most 1 nonzero entry in any row or column).  Second, we decompose 1-sparse Hamiltonians into Hamiltonians of the required form.  

\begin{restatable}{lemma}{ONESPARSE}
\label{lem:1sparse}
For any 1-sparse Hamiltonian $G$ and precision $\gamma>0$, there exist $O(\norm{G}_{\max}/\gamma)$ Hamiltonians $G_j$ with eigenvalues $\pm 1$ such that $\norm{{G} - \gamma \sum_{j}  G_j}_{\max} \leq \sqrt{2}\gamma$.
\end{restatable}

\begin{proof}
First we decompose the Hamiltonian $G$ as $G=G_X+iG_Y+G_Z$, where $G_X$ contains the off-diagonal real terms, $iG_Y$ contains the off-diagonal imaginary terms, and $G_Z$ contains the on-diagonal real terms.
Next, for each of $G_\xi$ for $\xi\in\{X,Y,Z\}$, we construct an approximation $\tilde G_\xi$ with each entry rounded off to the closest multiple of $2\gamma$.
Since each entry of $\tilde{G_\xi}$ is at most $\gamma$ away from the corresponding entry in $G_\xi$, we have $\norm{G_\xi-\tilde{G_\xi}}_{\max} \leq \gamma$.
Denoting $\tilde G=\tilde G_X+i\tilde G_Y+\tilde G_Z$, this implies $\norm{G-\tilde{G}}_{\max} \leq \sqrt{2}\gamma$.

Next, we take $C^\xi := \tilde G_\xi / \gamma$, so $\norm{C^\xi}_{\max} = \ceil{\norm{G_\xi}_{\max}/\gamma} \le \ceil{\norm{G}_{\max}/\gamma}$.
We can then decompose each 1-sparse matrix $C^\xi$ into $\|C^\xi\|_{\max}$ matrices, each of which is 1-sparse and has entries from $\{-2,0,2\}$.
If $C^\xi_{jk}$ is $2p$, then the first $|p|$ matrices in the decomposition have a $2$ for $p>0$ (or $-2$ if $p<0$) at the $(j,k)$ entry, and the rest have $0$.
More explicitly, we define
\begin{equation}
C^{\xi,\ell}_{jk} \defeq \begin{cases}
2 & \text{if } C^{\xi}_{jk} \ge 2\ell >0 \\
-2 & \text{if } C^{\xi}_{jk} \le -2\ell <0 \\
0 & \text{otherwise} \end{cases}
\end{equation}
for $\xi\in\{X,Y,Z\}$ and $\ell \in [\norm{C^{\xi}}_{\max}]$.
This gives a decomposition into at most $3 \ceil{\norm{G}_{\max}/\gamma}$ terms with eigenvalues in $\{-2,0,2\}$.

To obtain matrices with eigenvalues $\pm 1$, we perform one more step to remove the $0$ eigenvalues. We divide each $C^{\xi,\ell}$ into two copies, $C^{\xi,\ell,+}$ and $C^{\xi,\ell,-}$. For any column where $C^{\xi,\ell}$ is all zero, the corresponding diagonal element of $C^{\xi,\ell,+}$ is $+1$ (if $\xi \in \{X,Z\}$) or $+i$ (if $\xi=Y$) and the diagonal element of $C^{\xi,\ell,-}$ is $-1$ (if $\xi \in \{X,Z\}$) or $-i$ (if $\xi=Y$). Otherwise, we let $C^{\xi,\ell,+}_{jk}=C^{\xi,\ell,-}_{jk}=C^{\xi,\ell}_{jk}/2$. Thus $C^{\xi,\ell}=C^{\xi,\ell,+}+C^{\xi,\ell,-}$. Moreover, 
each column of $C^{\xi,\ell,\pm}$ has exactly one nonzero entry, which is $\pm 1$ (or $\pm i$ on the diagonal of $C^{Y,\ell,\pm}$).

This gives a decomposition $\tilde G/\gamma = \sum_{\ell,\pm} (C^{X,\ell,\pm}+iC^{Y,\ell,\pm}+C^{Z,\ell,\pm})$ in which each term has eigenvalues $\pm 1$.
The decomposition contains at most $6\ceil{\norm{G}_{\max}/\gamma} = O(\norm{G}_{\max}/\gamma)$ terms.
\end{proof}

\lem{1sparse} gives a decomposition of the required form as the eigenvalues can be adjusted to $0$ and $\pi$ by adding the identity matrix and multiplying by $\pi/2$.

It remains to decompose a sparse Hamiltonian into 1-sparse Hamiltonians.  Known results decompose a $d$-sparse Hamiltonian $H$ into a sum of $O(d^2)$ 1-sparse Hamiltonians~\cite{BAC+07}, but simulating one query to a $1$-sparse Hamiltonian requires $O(\log^* n)$ queries to the oracle for $H$.
We present a simplified decomposition theorem that decomposes a $d$-sparse Hamiltonian into $d^2$ 1-sparse Hamiltonians.
A query to the individual 1-sparse Hamiltonians can be performed using $O(1)$ queries to the original Hamiltonian, removing the $\log^* n$ factor.

\begin{restatable}{lemma}{DECOMPOSITION}
\label{lem:decomposition}
If $H$ is a $d$-sparse Hamiltonian, there exists a decomposition $H = \sum_{j=1}^{d^2} H_j$ where each $H_j$ is 1-sparse and a query to any $H_j$ can be simulated with $O(1)$ queries to $H$.
\end{restatable}

\begin{proof}
The new ingredient in our proof is to assume that the graph of $H$ is bipartite.  (Here the \emph{graph of $H$} has a vertex for each basis state and an edge between two vertices if the corresponding entry of $H$ is nonzero.) This is without loss of generality because we can simulate the Hamiltonian $\sigma_x \otimes H$ instead, which is indeed bipartite and has the same sparsity as $H$. From a simulation of $\sigma_x \otimes H$, we can recover a simulation of $H$ using the identity $e^{-i(\sigma_x \otimes H)t}|+\>|\psi\> =  |+\>e^{-iHt}|\psi\>$.

Now we decompose a bipartite $d$-sparse Hamiltonian into a sum of $d^2$ terms. To do this, we give an edge coloring of the graph of $H$ (i.e., an assignment of colors to the edges so that no two edges incident on the same vertex have the same color). Given such a coloring with $d^2$ colors, the Hamiltonian $H_j$ formed by only considering edges with color $j$ is 1-sparse.

We use the following simple coloring. For any pair of adjacent vertices $u$ and $v$, let $r(u,v)$ denote the rank of $v$ in $u$'s neighbor list, i.e., the position occupied by $v$ in a sorted list of $u$'s neighbors. This is a number between $1$ and $d$. Let the color of the edge $(u,v)$, where $u$ comes from the left part of the bipartition and $v$ comes from the right, be the ordered pair $(r(u,v),r(v,u))$.  This is a valid coloring since if $(u,v)$ and $(u,w)$ have the same color, then in particular the first component of the ordered pair is the same, so $r(u,v) = r(u,w)$ implies $v = w$. A similar argument handles the case where the common vertex is on the right.  

Given a color $(a,b)$, it is easy to simulate queries to the Hamiltonian corresponding to that color. To compute the nonzero entries of the $j$th row for this color,
if $j$ is in the left partition, then we find the neighbor of $j$ that has rank $a$; let us call this $\ell$.
Then we find the neighbor of $\ell$ that has rank $b$.
If this neighbor is $j$, then $\ell$ is the position of the nonzero entry in row $j$; otherwise there is no nonzero entry.
If $j$ is in the right partition, the procedure is the same, except with the roles of $a$ and $b$ reversed.
This procedure uses two queries.
\end{proof}

Observe that the simple trick of making the Hamiltonian bipartite suffices to remove the $O(\log^* n)$ term present in previous decompositions of this form. This trick is quite general and can be applied to remove a factor of $O(\log^* n)$ wherever such a factor appears in a known Hamiltonian simulation algorithm (e.g., \cite{BAC+07,CK11,WBHS11}).

\lem{decomposition} decomposes our Hamiltonian $H$ into $d^2$ 1-sparse Hamiltonians. We further decompose $H$ using \lem{1sparse} into a sum of $\eta = O(d^2\norm{H}_{\max}/\gamma)$ Hamiltonians $G_j$ such that $\norm{H - \gamma \sum_{j=1}^{\eta} G_j}_{\max} \leq \sqrt{2}\gamma d^2$, since each 1-sparse Hamiltonian is approximated with precision $\sqrt{2}\gamma$ and there are $d^2$ approximations in this sum. To upper bound the simulation error, we have
$
	\norm{e^{-iHt} - e^{-i\gamma \sum_j G_j t}} 
	\le \norm{(H - \gamma \sum_{j=1}^{\eta} G_j)t}
	\le \sqrt{2} \gamma d^3 t
$,
where we used the fact that $\norm{e^{iA} - e^{iB}} \leq \norm{A-B}$ (as explained in the proof of \thm{equiv}) and $\norm{A} \leq d \norm{A}_{\max}$ for a $d$-sparse matrix $A$. Choosing  $\gamma = {\epsilon}/{\sqrt{2} d^3 t}$ gives the required precision.  We now invoke \cor{hamsim} with number of Hamiltonians $\eta = O(d^2\norm{H}_{\max}/\gamma)$ and simulation time $\gamma t$ to get $\tau =  d^2\norm{H}_{\max}t$.
Plugging this value of $\tau$ into \cor{hamsim} gives us the following lemma, which is the query-complexity part of \thm{sparse}.

\begin{lemma}
\label{lem:sparse}
A $d$-sparse Hamiltonian $H$ can be simulated for time $t$ with error at most $\epsilon$ using $O\big(\tau \frac{\log(\tau/\epsilon)}{\log\log(\tau/\epsilon)}\big)$ queries, where $\tau \defeq d^2 \norm{H}_{\max} t \ge 1$.
\end{lemma}

Note that above we have determined the values of $r$ and $\gamma$ to use, but these values do not affect the query complexity (although they do affect the time complexity). This is because $r$ and $\gamma$ affect the value of $m$, but the analysis in \sec{main} is independent of $m$. This enables a simple generalization to time-dependent Hamiltonians. We can approximate the true evolution by a product of evolutions under time-independent Hamiltonians for each of the $r$ time intervals of length $t/r$. Provided the derivative of the Hamiltonian is bounded, this approximation can be made arbitrarily accurate by choosing $r$ large enough. As the query complexity does not depend on $r$, it is independent of $h'$, similar to \cite{PQSV11}.

Finally, consider simulating a $k$-local Hamiltonian.  A term acting nontrivially on at most $k$ qubits is $2^k$-sparse: two states $x,y \in \{0,1\}^n$ are adjacent if the only bits on which $x$ and $y$ differ are among the $k$ bits involved in the local term.  Using this structure, we can give an explicit $2^k$-coloring, improving over the $4^k$-coloring provided by \lem{decomposition}: we simply color an edge between states $x$ and $y$ by indicating which of the $k$ bits are flipped.  Thus we can decompose a $k$-local Hamiltonian with $M$ terms as a sum of $2^k M$ 1-sparse Hamiltonians.  Using this decomposition in place of \lem{decomposition}, we find a simulation as in \thm{sparse} but with $\tau$ replaced by $\tilde\tau \defeq 2^k M \norm{H}_{\max} t$.

\section{Time complexity}
\label{sec:time}

We now consider the time complexities of the algorithms described in \thm{cquerysim} and \thm{sparse} (recall that time complexity refers to the sum of the number of queries and additional 2-qubit gates used in the algorithm).  Our approach considerably simplifies this analysis over previous work and gives improved upper bounds.

The basic algorithm as described in \sec{main} is inefficient as it relies on creating a state of $m = \poly(h,T,\frac{1}{\epsilon})$ qubits.  Instead, as in previous work \cite{BCG14}, we create a compressed version of this state that allows us to perform the necessary controlled operations and to reflect about the zero state.  Our simplified approach does not require measuring the control qubits, an operation that accounts for much of the technical complexity of \cite{BCG14}.

We now prove \thm{cquerysim} from \sec{intro}, which we restate for convenience.

\CQUERYSIM*

\begin{proof}
The query complexity of this theorem was established in \lem{cquerysim}.
As in the analysis of query complexity, it suffices to simulate a segment implementing evolution for time $1/5$ with precision $\epsilon/5T$.  To simulate the continuous-query model, we can assume without loss of generality that query evolutions are approximated (as in \thm{equiv}) by $m$ fractional evolutions of equal length $1/5m$.  Thus we can assume that in each segment, as defined in \lem{segment}, $\alpha \defeq \alpha_i=1/5m$ for all $i \in [m]$.  Let $c \defeq \cos(\pi/10m)$ and $s \defeq \sin(\pi/10m)$.

The idealized initial state of the ancilla qubits (i.e., the state in the dotted box of \fig{segment}) is
\be
  \left(\frac{\sqrt{c}|0\> + \sqrt{s}|1\>}{\sqrt{c+s}}\right)^{\otimes m}
  = \sum_{\x \in \{0,1\}^m} \kappa^{m-|\x|} \sigma^{|\x|} |\x\>,
\ee
where $\kappa \defeq \frac{\sqrt{c}}{\sqrt{c+s}}$ and $\sigma \defeq \frac{\sqrt{s}}{\sqrt{c+s}}$.  We truncate this state to the subspace of those $\x$ with Hamming weight $|\x| \le k$.  Specifically, we prepare the encoded state
\be
  \sum_{\ell \in L} \kappa^{m-|\ell|} \sigma^{|\ell|} |\ell\> 
  + \delta |{\perp}\>,
\label{eq:encodedstate}
\ee
where $L \defeq \{(\ell_1,\ldots,\ell_h)\colon 1 \le h \le k,\, \ell_1+\cdots+\ell_h \le m-h\}$,
$|{\perp}\>$ is a special state orthogonal to all terms in the first sum, and the coefficient $\delta$ was shown to be small in \lem{approxsegment}.  Observe that there is a natural  bijection between $L$ and the set of strings $\x$ with $|\x| \le k$, given by $\x \leftrightarrow 0^{\ell_1}10^{\ell_2}10^{\ell_3}\ldots0^{\ell_h}10^{m-h-\ell_1-\cdots-\ell_h}$.

It is straightforward to perform the operation \eq{rearranged} from the proof of \lem{approxsegment}, conditioning on $\x$ as represented by $\ell$.  Recall that $W_i(\x)$ represents the evolution under the driving Hamiltonian from time $\sum_{j=1}^i \ell_j/5m$ to time $\sum_{j=1}^{i+1}\ell_{j}/5m$ (where we define $\ell_{k+1}=m$).  By assumption, any such evolution can be performed with precision $O(\epsilon/T)$ using $g$ gates.  Also, recall that $Q_i(\x)$ is simply $Q$ if $i \le |\x|$ or $\id$ otherwise, so it can be applied in time $O(\log k)$.  Thus the operation \eq{rearranged} can be applied in time $O(k(g+\log k))$.

At the end of the segment we must effectively apply the final $P$ and $R$ gates to the encoded state before reflecting about the encoding of $|0^m\>$.
(That is, we jointly reflect about this state and $|0\>$ for the additional ancilla in \fig{segment}.)
The $P$ gates are straightforward to apply in the given encoding.  Rather than apply the encoded $R$ gates directly, reflect about the encoding of $|0^m\>$, and then apply the encoded $R$ gates for the next segment, it suffices to reflect about the encoding of $R_\alpha^{\otimes m}|0^m\>$ (note that $R_\alpha^\dag = R_\alpha$). This can be done by applying the inverse of the procedure for preparing \eq{encodedstate}, reflecting about the initial state, and applying the preparation procedure. Overall, we see that the segment can be applied to the encoded initial state with suitable accuracy using $O(k(g+\log m))$ gates, plus the cost of preparing the encoded ancillas.

The encoded initial state \eq{encodedstate} can be prepared in time $O(k(\log m + \log\log(1/\epsilon))) = O(k \log m)$, as described in Sections 4.2--4.4 of \cite{BCG14} (see in particular equation (22)).
Since $k = O\big(\frac{\log(T/\epsilon)}{\log\log(T/\epsilon)}\big)$ (from the proof of \lem{approxsegment} with error at most $\epsilon/5T$) and $m = \poly(T,\h,\frac{1}{\epsilon})$ (from \thm{equiv}), the overall complexity of making the encoded ancilla state is $O\big(\frac{\log(T/\epsilon) \log(\h T/\epsilon)}{\log\log(T/\epsilon)}\big)$. Thus the cost of implementing a constant-query algorithm to precision $\epsilon/5T$ is \be
  O(k(g+\log m))
  =O\left(\frac{\log(T/\epsilon)}{\log\log(T/\epsilon)}
          [g+\log(\h T/\epsilon)]\right).
\ee

Implementing $O(T)$ segments, each with this complexity, gives the stated time complexity.
With error bounded by $\epsilon/5T$ for each segment, the overall error is at most $\epsilon$.
\end{proof}

Using this approach we can similarly prove \thm{sparse} from \sec{intro}, which we restate for convenience.

\SPARSE*

\begin{proof}
The query complexity of this theorem was established in \lem{sparse}.
Since the query complexity of \thm{sparse} is proved by reduction to \thm{cquerysim}, a time-efficient version of \thm{sparse} can be obtained by essentially the same procedure as the time-efficient version of \thm{cquerysim}.  In this reduction, $\tau$ plays the role of $T$.
Note that the reduction ultimately uses a fractional-query simulation, so we cannot directly use the result as stated in \thm{cquerysim}, where the time-complexity is for the continuous-query case.
Nevertheless, we can obtain a similar result if $g$ is taken to represent the cost of performing any sequence of consecutive non-query operations in the fractional-query algorithm.
The term $\log(\h T/\epsilon)$ in \thm{cquerysim} results from discretizing a continuous-query algorithm with a driving Hamiltonian and does not arise here.

The non-query operations $V_j$ for $j \in [m]$ described in the proof of \thm{multiquery} are straightforward to implement.  In the application to Hamiltonian simulation, we simply cycle through all $\eta$ terms in order, so all the $V_j$s can simply add $1$ modulo $\eta$, and a sequence $V_{j'}\cdots V_j$ adds $j'-j \bmod \eta$.  Without loss of generality, we can assume $\eta$ is a power of $2$, so addition modulo $\eta$ can be performed by standard binary addition, keeping only the $\log_2\eta$ least significant bits.  Thus any operation to be performed between queries can be applied using $g = O(\log \eta) = O(\log(d\norm{H}_{\max}t/\epsilon))$ operations (where the value of $\eta$ is discussed following the proof of \lem{decomposition}). Next, observe that it suffices to decompose the evolution into $m = \eta^3 t^2/\epsilon = \poly(t,\norm{H}_{\max},d,\frac{1}{\epsilon})$ terms (as stated in \cor{hamsim}).  In the proof of \thm{cquerysim}, the time complexity for a constant-query algorithm is $O(k(g+\log m))$.  This upper bounds the number of additional gates required to perform the non-query operations.
Using $g= O(\log(d\norm{H}_{\max}t/\epsilon))$ and $\log m = O(\log(d\norm{H}_{\max}t/\epsilon))$, we see that this is $O\big(\tau \frac{\log^2(\tau/\epsilon)}{\log\log(\tau/\epsilon)}\big)$.

This only accounts for the operations performed between applications of the unitary $Q$ defined in \cor{hamsim}. It remains to implement $Q \defeq  \sum_{j=1}^\eta |j\>\<j| \otimes e^{-iH_j}$ using the oracle, where $H = \sum_{j=1}^\eta H_j$ and $H_j$ are Hamiltonians with eigenvalues $0$ and $\pi$. To implement $Q$ we need to read the first register to learn which 1-sparse Hamiltonian is to be simulated and then simulate the 1-sparse Hamiltonian $H_j$. The first part is straightforward; from $j$ we can determine which 1-sparse Hamiltonian is to be simulated and whether it is an $X$, $Y$, or $Z$ term, in the notation of \lem{1sparse}. This can be done with $O(\log \eta)$ gates, which is linear in the size of the first register. Now we need to implement the 1-sparse Hamiltonian on an $n$-qubit register. This can be done with $O(n)$ gates using the constructions in \cite{AT03,CCDFGS03}. For example, to implement an $X$ Hamiltonian on a state $|v\>$, we can write down the index of $v$'s neighbor in another register, swap the two registers, and uncompute the second register. Thus we can implement $Q$ using $O(\log \eta + n)$ gates. Since the number of uses of $Q$ is the query complexity, the total number of gates used for all invocations of $Q$ and the non-query operations is $O\big(\tau \frac{\log(\tau/\epsilon)}{\log\log(\tau/\epsilon)}[\log(\tau/\epsilon)+n]\big)$, which is $O\big(\tau \frac{\log^2(\tau/\epsilon)}{\log\log(\tau/\epsilon)}n\big)$.
\end{proof}

The same techniques can be straightforwardly applied to simulate time-dependent sparse Hamiltonians. We divide the evolution into intervals of length $t/r$, so the Hamiltonian can change by no more than $h't/r$ over such an interval, where $h' \defeq \max_{s \in [0,t]} \norm{\frac{\mathrm{d}}{\mathrm{d}s} H(s)}$. Thus the error for each interval is $O(h't^2/r^2)$, and the error in the overall simulation is $O(h't^2/r)$. Therefore it suffices to take $r=\Omega(h't^2/\epsilon)$. Then $m = \poly(t,h,h',d,\frac{1}{\epsilon})$, and the complexity is
$O\big(\tau \frac{\log(\tau/\epsilon) \log((\tau + \tau')/\epsilon)}{\log\log(\tau/\epsilon)}n\big)$
as stated.

\section{Lower bounds}
\label{sec:lb}

We now show that in general, any sparse Hamiltonian simulation method must use $\Omega\big(\frac{\log(1/\epsilon)}{\log\log(1/\epsilon)}\big)$ discrete queries to obtain error at most $\epsilon$, so dependence of the query complexity in \thm{sparse} on $\epsilon$ is tight up to constant factors.  To show this, we use ideas from the proof of the no-fast-forwarding theorem~\cite[Theorem 3]{BAC+07}, which says that generic Hamiltonians cannot be simulated in time sub-linear in the evolution time. The Hamiltonian used in the proof of that theorem has the property that simulating it for time $t = \pi n/2$ determines the parity of $n$ bits exactly.  We observe that simulating this Hamiltonian (with sufficiently high precision) for any time $t > 0$ gives an unbounded-error algorithm for the parity of $n$ bits, which also requires $\Omega(n)$ queries~\cite{FGGS98,BBC+01}.

We now prove \thm{hamsimlower} from \sec{intro}, which we restate for convenience.

\LOWERBOUND*

\begin{proof}
To construct the Hamiltonian, we begin with a simpler Hamiltonian $H'$ that acts on vectors $|i\>$ with $i \in \{0,1,\ldots,N\}$ \cite{CDEL04}. The nonzero matrix entries of $H'$ are $\<i\left|H'\right|i+1\>=\<i+1\left|H'\right|i\>=\sqrt{(N-i)(i+1)}/N$ for $i \in \{0,1,\ldots,N-1\}$. We have $\norm{H'}_{\max} < 1$, and simulating $H'$ for $t = \pi N/2$ starting with the state $|0\>$ gives the state $|N\>$ (i.e., $e^{-iH'\pi N/2}|0\>=|N\>$). More generally, for $t \in [0,\pi N/2]$, we claim that $|\<N|e^{-iH't}|0\>| = |{\sin(t/N)}|^N$.

To see this, consider the Hamiltonian $\bar X \defeq \sum_{j=1}^N X^{(j)}$, where $X \defeq \left(\begin{smallmatrix}0 & 1 \\ 1 & 0\end{smallmatrix}\right)$ and the superscript $(j)$ indicates that the operator acts nontrivially on the $j$th qubit.  Since $e^{-iXt} = \cos(t)\id - i \sin(t) X$, we have $|\<11\ldots 1|e^{-i\bar Xt}|00\ldots 0\>| = |{\sin(t)}|^N$.  Defining $|{\wt_k}\> \defeq \binom{N}{k}^{-1/2} \sum_{|x|=k} |x\>$,
we have 
\be
\bar X|{\wt_k}\>
= \sqrt{(N-k+1)k} |{\wt_{k-1}}\> + \sqrt{(N-k)(k+1)} |{\wt_{k+1}}\>.
\ee
This is precisely the behavior of $NH'$ with $|k\>$ playing the role of $|{\wt_k}\>$, so the claim follows.

Now, as in \cite{BAC+07}, consider a Hamiltonian $H$ generated from an $N$-bit string $x_1 x_2 \ldots x_{N}$. $H$ acts on vertices $|i,j\>$ with $i\in \{0,\ldots,N\}$ and $j\in \{0,1\}$. The nonzero matrix entries of this Hamiltonian are  
\be
\<i,j\left|H\right|i-1,j\oplus x_i\>=\<i-1,j\oplus x_i\left|H\right|i,j\>=\sqrt{(N-i+1)i}/N
\ee
for all $i$ and $j$. By construction, $|0,0\>$ is connected to either $|i,0\>$ or $|i,1\>$ (but not both) for any $i$; it is connected to $|i,j\>$ if and only if $j=x_1 \oplus x_2 \oplus \cdots \oplus x_{i}$. Thus $|0,0\>$ is connected to either $|N,0\>$ or $|N,1\>$, and determining which is the case determines the parity of $x$. The graph of this Hamiltonian 
contains two disjoint paths, one containing $|0,0\>$ and $|N,\parity(x)\>$ and the other containing $|0,1\>$ and $|N,1 \oplus {\parity(x)}\>$. Restricted to the connected component of $|0,0\>$, this Hamiltonian is the same as $H'$. Thus, starting with the state $|0,0\>$ and simulating $H$ for time $t$ gives $|\<N,\parity(x)|e^{-iHt}|0,0\>| = |{\sin(t/N)}|^N$. Furthermore, for any $t$, we have $\<N,1 \oplus {\parity(x)}|e^{-iHt}|0,0\> = 0$ since the two states lie in disconnected components.

Simulating this Hamiltonian exactly for any time $t>0$ starting with $|0,0\>$ yields an unbounded-error algorithm for computing the parity of $x$, as follows. First we measure $e^{-iHt}|0,0\>$ in the computational basis.  We know that for any $t>0$, the state $e^{-iHt}|0,0\>$ has some nonzero overlap on $|N,\parity(x)\>$ and zero overlap on $|N,1\oplus\parity(x)\>$. If the first register is not $N$, we output 0 or 1 with equal probability. If the first register is $N$, we output the value of the second register. This is an unbounded-error algorithm for the parity of $x$, and thus requires $\Omega(N)$ queries.

Since the unbounded-error query complexity of parity is $\Omega(N)$ \cite{FGGS98,BBC+01}, this shows that exactly simulating $H$ for any time $t>0$ needs $\Omega(N)$ queries. However, even if we only have an approximate simulation, the previous algorithm still works as long as the error in the output state is smaller than the overlap $|\<N,\parity(x)|e^{-iHt}|0,0\>|$. If we ensure that the overlap is larger than $\epsilon$ by a constant factor, then even with error $\epsilon$, the overlap on that state will be larger than $\epsilon$. On the other hand, the overlap on $|N,1\oplus\parity(x)\>$ is at most $\epsilon$, since the output state is $\epsilon$ close to the ideal output state which has no overlap. 

To achieve an overlap much larger than $\epsilon$, we need $|{\sin(t/N)}|^{N}$ to be much larger than $\epsilon$. There is some value of $N$ in $\Theta\big(\frac{\log(1/\epsilon)}{\log\log(1/\epsilon)}\big)$ that achieves this.
\end{proof}

A similar construction shows that any $\epsilon$-error simulation of the continuous-query model must use $\Omega\big(\frac{\log(1/\epsilon)}{\log\log(1/\epsilon)}\big)$ discrete queries, so \lem{main} is tight up to constant factors.  Again we show that a sufficiently high-precision simulation of a certain Hamiltonian could be used to compute parity with unbounded error.  However, in the fractional-query model, the form of the Hamiltonian is restricted and it is unclear how to implement the weights that simplify the analysis of the dynamics in \thm{hamsimlower}.  Instead, we consider a quantum walk on an infinite unweighted path that also solves parity with unbounded error, and we show that this still holds if the path is long but finite.

\begin{theorem}[$\epsilon$-dependent lower bound for continuous-query simulation]
\label{thm:fqsimlower}
For any $\epsilon>0$, given a query Hamiltonian $H_x$ for a string of $N = \Theta\big(\frac{\log(1/\epsilon)}{\log\log(1/\epsilon)}\big)$ bits, simulating $H_x + H_D(t)$ for constant time  with precision $\epsilon$ requires $\Omega(N)$ queries.
\end{theorem}

\begin{proof}
We prove a lower bound for simulating a Hamiltonian of the form $H'=\sum_{a=1}^\eta c_a U_a^\dag H_x U_a$ with coefficients $c_1,\ldots,c_\eta \in \R$.  The Hamiltonian $H_x$ can be used to simulate $H'$ to any given accuracy with overhead $\sum_a |c_a|$, so this implies a lower bound for simulating $H_x$.  In particular, by taking $r$ sufficiently large, the evolution under $H'$ can be approximated arbitrarily closely as
\begin{equation}
e^{-iH't} \approx \left(\prod_{a=1}^\eta U_a^\dag e^{-i H_x c_a t/r} U_a\right)^r.
\end{equation}
This corresponds to a fractional-query algorithm with cost $t \sum_{a=1}^\eta |c_a|$. By \thm{equiv}, this fractional-query algorithm can be simulated with arbitrarily small error by a continuous-query algorithm with the same cost. This continuous-query algorithm uses the query Hamiltonian $H_x$, and its driving Hamiltonian $H_D(t)$ implements the unitaries $\{U_a,U_a^\dag\}_{a=1}^\eta$ at appropriate times.

Viewing the Hamiltonian in terms of the graph of its nonzero entries, the oracle Hamiltonian $H_x$ provides input-dependent self-loops.  First we modify it to give input-dependent edges.  Observe that
\begin{align}
  \Had \begin{pmatrix}1&0\\0&0\end{pmatrix} \Had
  &= \frac{1}{2}\begin{pmatrix}1&1\\1&1\end{pmatrix}
\end{align}
where $\Had \defeq \frac{1}{\sqrt2}(\begin{smallmatrix}1 & 1\\1 & -1\end{smallmatrix})$ is the Hadamard gate.  Thus we can include a term in the Hamiltonian that has an edge between two vertices associated with the input index $i$ (and self-loops on those vertices) if $x_i=1$, and is zero otherwise.

Now consider a space with basis states $|i,j,k\>$ where $i \in \Z$ and $j,k \in \{0,1\}$.  The label $j$ plays the same role as in \thm{hamsimlower}, whereas the new label $k$ indexes two positions for each value of $i$.  These new positions are needed because the pairs of vertices associated with each input index must be disjoint.

To specify the Hamiltonian, we define unitaries $U_1,U_2,U_3,U_4$ so that the nonzero matrix elements of $U_a^\dag H_x U_a$ for $a \in \{1,2,3,4\}$ are
\begin{align}
     \<i,0,k|U_1^\dag H_x U_1|i,0,\bar k\> 
  &= \<i,0,k|U_1^\dag H_x U_1|i,0,k\> = x_i/2 \\
     \<i,1,k|U_2^\dag H_x U_2|i,1,\bar k\> 
  &= \<i,1,k|U_2^\dag H_x U_2|i,1,k\> = x_i/2 \\
     \<i,k,k|U_3^\dag H_x U_3|i,\bar k,\bar k\> 
  &= \<i,k,k|U_3^\dag H_x U_3|i,k,k\> = x_i/2 \\
     \<i,k,\bar k|U_4^\dag H_x U_4|i,\bar k,k\> 
  &= \<i,\bar k,k|U_4^\dag H_x U_4|i,\bar k,k\> = x_i/2
\end{align}
for all $i \in [N]$ and $k \in \{0,1\}$.  Combining these four contributions to obtain a Hamiltonian $- U_1^\dag H_x U_1 - U_2^\dag H_x U_2 + U_3^\dag H_x U_3 + U_4^\dag H_x U_4$ and observing that the self-loops cancel, these matrix elements can be summarized in terms of the gadget shown in \fig{paritygadget}.

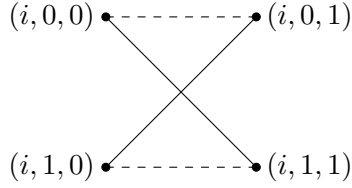
\begin{figure}[ht]
\begin{center}
\begin{tikzpicture}
\coordinate (tl) at (0,2);
\coordinate (tr) at (2,2);
\coordinate (bl) at (0,0);
\coordinate (br) at (2,0);

\filldraw (tl) circle (1.5pt) node [left] {$(i,0,0)$};
\filldraw (tr) circle (1.5pt) node [right] {$(i,0,1)$};
\filldraw (bl) circle (1.5pt) node [left] {$(i,1,0)$};
\filldraw (br) circle (1.5pt) node [right] {$(i,1,1)$};

\draw[dashed] (tl) -- (tr);
\draw[dashed] (bl) -- (br);
\draw (tl) -- (br);
\draw (bl) -- (tr);
\end{tikzpicture}
\end{center}
\vspace{-1em}
\caption{\label{fig:paritygadget}
The gadget for querying $x_i$.  If $x_i=0$, no edges are present.  If $x_i=1$, the solid edges have weight $1/2$ and the dashed edges have weight $-1/2$.}
\end{figure}

We add a driving Hamiltonian to connect these gadgets to form two paths encoding the parity similarly as in \thm{hamsimlower}, and we extend the paths infinitely in both directions.  Specifically, the driving Hamiltonian $H_D$ has nonzero matrix elements
\begin{align}
   \<i,j,k|H_D|i,j,\bar k\> &= 1/2
\end{align}
for all $i \in \Z$ and $j,k \in \{0,1\}$ (corresponding to the dashed edges in \fig{paritygadget}, but with positive weight), and
\begin{align}
  \<i+1,j,0|H_D|i,j,1\> &= \<i,j,1|H_D|i+1,j,0\> = 1/2
\end{align}
for all $i \in \Z$ and $j \in \{0,1\}$ (corresponding to edges that join sectors with adjacent values of $i$). Then the total Hamiltonian
\be
  H = - U_1^\dag H_x U_1 - U_2^\dag H_x U_2 + U_3^\dag H_x U_3 + U_4^\dag H_x U_4 + H_D
\label{eq:fraclbham}
\ee
is $1/2$ times the adjacency matrix of the disjoint union of two infinite paths, one with vertices
\be
  \begin{aligned}
  \ldots,&(0,0,0),(0,0,1),(1,0,0),(1,x_1,1),(2,x_1,0),(2,x_1 \oplus x_2,1),\ldots, \\
  &(N,x_1 \oplus \cdots \oplus x_N,1),(N+1,x_1 \oplus \cdots \oplus x_N,0),(N+1,x_1 \oplus \cdots \oplus x_N,1),\ldots
  \end{aligned}
\ee
and the other with vertices
\be
  \begin{aligned}
  \ldots,&(0,1,0),(0,1,1),(1,1,0),(1,1 \oplus x_1,1),(2,1 \oplus x_1,0),(2,1 \oplus x_1 \oplus x_2,1),\ldots, \\
  &(N,1 \oplus x_1 \oplus \cdots \oplus x_N,1),(N+1,1 \oplus x_1 \oplus \cdots \oplus x_N,0),(N+1,1 \oplus x_1 \oplus \cdots \oplus x_N,1),\ldots.
  \end{aligned}
\ee
Analogous to the Hamiltonian $H$ in the proof of \thm{hamsimlower}, $(0,0,1)$ is in the same component as $(N,b,1)$ if and only if $b=\parity(x)$.

To compute the probability of reaching $(n,\parity(x),1)$ starting from $(0,0,1)$ after evolving with the Hamiltonian \eq{fraclbham} for time $t$, we can use the expression for the propagator on an infinite path in terms of a Bessel function (see for example \cite{CCDFGS03}).  Specifically, we have
\be
  |\<N,\parity(x),1|e^{-iHt}|0,0,1\>|
  = |J_{2N}(t)|.
\ee
For large $N$ and for any fixed $t \ne 0$, we have $|J_N(t)| = e^{-\Theta(N\log N)}$ \cite[Section 8.1]{Wat22}.  Thus, as in the proof of \thm{hamsimlower}, even a simulation with error $\epsilon$ gives the result with nonzero probability provided $N = \Theta\big(\frac{\log(1/\epsilon)}{\log\log(1/\epsilon)}\big)$.

The preceding argument uses a Hamiltonian acting on an infinite-dimensional space.  However, we can truncate it to act on a finite space with essentially the same effect.  Specifically, we apply the Truncation Lemma of \cite{CGW13} with $\mathcal{K} = \spn\{|i,j,k\> \colon -N^3-N^2 \le i \le N^3+N^2,\, j,k \in \{0,1\}\}$ and $W=H$.  Let $P$ project onto $\mathcal{K}$ and let $P'$ project onto $\spn\{|i,j,k\> \colon -N^2 \le i \le N^2,\, j,k \in \{0,1\}\}$.  Finally, let $|\gamma(t)\> = P'e^{-iHt}|0,0,1\>$.  Then $\delta^2 \defeq \norm{e^{-iHt}|0,0,1\> - |\gamma(t)\>}^2 = |J_{2N^2+1}(t)|^2+2\sum_{j=2}^\infty |J_{2N^2+j}(t)|^2 \le e^{-\Omega(N^2 \log N)}$.  Furthermore, $(\id-P)H^r|\gamma(t)\>=0$ for all $r \in \{0,1,\ldots,N^3\}$.  Also observe that $\norm{H}=1$.  Thus the Truncation Lemma shows that
\be
  \norm{(e^{-iHt} - e^{-iPHPt}) |0,0,1\>}
  \le \left( \frac{4et}{N^3}+2 \right)\bigl(\delta+2^{-N^3}(1+\delta)\bigr)
  \le e^{-\Omega(N^2 \log N)},
\ee
so the error incurred by truncating $H$ to the Hamiltonian $PHP$ acting on the finite-dimensional space $\mathcal{K}$ is asymptotically negligible compared to $\epsilon$.
\end{proof}

\section{Open questions}
\label{sec:conclusion}

While our algorithm for continuous-query simulation is optimal as a function of $\epsilon$ alone, it is suboptimal as a function of $T$, and it is unclear what tradeoffs might exist between these two parameters.  The best known lower bound as a function of both $\epsilon$ and $T$ is $\Omega\big(T+\frac{\log(1/\epsilon)}{\log\log(1/\epsilon)}\big)$.  It would be surprising if this bound were achievable, but it remains open to find such an algorithm or to prove a better lower bound.  In general, any improvement to the tradeoff between $\epsilon$ and $T$ could be of interest.

In the context of time-independent sparse Hamiltonian simulation, the quantum walk-based simulation of \cite{Chi10,BC12} achieves linear dependence on $t$, whereas our upper bound is superlinear in $t$.  However, the dependence on $\epsilon$ is significantly worse in the walk-based approach.  It would be desirable to combine the benefits of these two approaches into a single algorithm.

Another open question is to better understand the dependence of our sparse Hamiltonian simulation method on the sparsity $d$.  While we use $d^{2+o(1)}$ queries, the method of \cite{BC12} uses only $O(d)$ queries.  Could the performance of the simulation based on fractional queries be improved by a different decomposition of the Hamiltonian?

\section*{Acknowledgements}

We thank Sevag Gharibian and Nathan Wiebe for valuable discussions.
This work was supported in part by ARC grant FT100100761, Canada's NSERC, CIFAR, the Ontario Ministry of Research and Innovation, and the US ARO.  RDS acknowledges support from the Laboratory Directed Research and Development Program at Los Alamos National Laboratory.


\providecommand{\bysame}{\leavevmode\hbox to3em{\hrulefill}\thinspace}


\appendix


\section{Proofs of known results}
\label{app:proofs}

In this appendix, for the sake of completeness we provide proofs of claims that are known or essentially follow from known results.


\subsection{Equivalence of continuous- and fractional-query models}
\label{app:equiv}

\EQUIV* 

\begin{proof}
A simulation of the continuous-query model by the fractional-query model with the stated properties appears in Section II.A of \cite{CGM+09}.  We present their proof for completeness.

We wish to implement the unitary $U(T)$ satisfying the Schr\"{o}dinger equation \eq{schrodinger} with $U(0)=\id$. To refer to the solutions of this equation for arbitrary Hamiltonians and time intervals, we define $U_H(t_2,t_1)$ to be the solution of the Schr\"{o}dinger equation with Hamiltonian $H$ from time $t_1$ to time $t_2$ where $U(t_1) = \id$.  In this notation, $U(T) = U_{H_x + H_D}(T,0)$.

Let $m$ be an integer and $\theta = T/m$. We have
\be
  U_{H_x + H_D}(T,0)
  = U_{H_x + H_D}(m\theta,(m-1)\theta) \cdots 
    U_{H_x + H_D}(2\theta,\theta) U_{H_x + H_D}(\theta,0).
\ee
If we can approximate each of these $m$ terms, we can use the subadditivity of error in implementing unitaries (i.e., $\norm{UV - \tilde{U}\tilde{V}} \leq \norm{U - \tilde{U}} + \norm{V - \tilde{V}}$ for unitaries $U, \tilde{U}, V, \tilde{V}$) to obtain an approximation of $U(T)$.

Reference \cite{HR90} shows that for small $\theta$, the evolution according to Hamiltonians $A$ and $B$ over an interval of length $\theta$ approximates the evolution according to $A+B$ over the same interval.  Specifically, from \cite[eq.~A8b]{HR90} we have
\be 
\norm{U_{A+B}((j+1)\theta,j\theta) - U_{A}((j+1)\theta,j\theta) U_{B}((j+1)\theta,j\theta)} \leq \int_{j\theta}^{(j+1)\theta} \mathrm{d}v \int_{j\theta}^{v}\mathrm{d}u \, \norm{[A(u),B(v)]}.
\ee
In our application, $A(t) = H_D(t)$ and $B=H_x$. Since $\norm{H_x}=1$, the right-hand side is at most 
\be
2 \int_{j\theta}^{(j+1)\theta} \mathrm{d}v \int_{j\theta}^{v}  \mathrm{d}u \, \norm{H_D(u)} \leq 
2 \int_{j\theta}^{(j+1)\theta} \mathrm{d}v \int_{j\theta}^{(j+1)\theta}  \mathrm{d}u \, \norm{H_D(u)} =
2\theta \int_{j\theta}^{(j+1)\theta} \norm{H_D(u)} \, \mathrm{d}u.
\ee

By subadditivity, the error in implementing $U(T)$ is at most
\be
2\theta \sum_{j=0}^{m-1} \int_{j\theta}^{(j+1)\theta} \norm{H_D(u)} \, \mathrm{d}u 
= 2\theta  \int_{0}^{T} \norm{H_D(u)} \mathrm{d}u 
= 2\theta \h T 
= \frac{2\h T^2}{m}.
\ee
This error is smaller than $\epsilon$ when $m \geq 2\h T^2/\epsilon$, which proves this direction of the equivalence.

For the other direction, consider a fractional-query algorithm
\be
  U_{\mathrm{fq}} \defeq U_m Q^{\alpha_m} U_{m-1} \cdots Q^{\alpha_2} U_1 Q^{\alpha_1} U_0
\label{eq:fqcircuit}
\ee
(recall that $Q$ depends on $x$),
where $\alpha_i \in (0,1]$ for all $i \in [m]$, with complexity $T = \sum_{i=1}^m \alpha_i$.  Let $A_i \defeq \sum_{j=1}^i \alpha_j$ for all $i \in [m]$ and let $U_j \equalscolon e^{-i H_D^{(j)}}$ for all $j \in \{0,1,\ldots,m\}$.  Consider the piecewise constant Hamiltonian
\be
  H(t) = H_x + \frac{1}{\epsilon_1} \left(\delta_{t \in [0,\epsilon_1]} H_D^{(0)} + \sum_{i=1}^m \delta_{t \in [A_i-\epsilon_1,A_i]} H_D^{(i)} \right),
\ee
where $\delta_B$ is $0$ if $B$ is false and $1$ if $B$ is true.  Provided $\epsilon_1 \le \min\{\alpha_1/2,\alpha_2,\ldots,\alpha_m\}$, evolving with $H(t)$ from $t=0$ to $T$ implements a unitary close to our fractional-query algorithm. More precisely, it implements
\be
  \begin{aligned}
  U(T) &=
  e^{-i (H_D^{(m)} + \epsilon_1 H_x)} e^{-i(\alpha_m - \epsilon_1) H_x} 
  e^{-i (H_D^{(m-1)} + \epsilon_1 H_x)} \cdots \\
  &~\quad e^{-i(\alpha_2 - \epsilon_1) H_x} 
  e^{-i (H_D^{(1)} + \epsilon_1 H_x)} e^{-i(\alpha_1 - 2\epsilon_1) H_x} 
  e^{-i (H_D^0 + \epsilon_1 H_x)},
  \end{aligned}
\label{eq:piecewise}
\ee
which satisfies $\norm{U(T) - U_{\mathrm{fq}}} = O(m\epsilon_1)$. This follows from the fact that each exponential in \eq{piecewise} approximates the corresponding unitary of \eq{fqcircuit} within error $\epsilon_1$ (e.g., $\norm{e^{-i (H_D^{(m)} + \epsilon_1 H_x)} - U_m} = O(\epsilon_1)$ and $\norm{e^{-i(\alpha_m - \epsilon_1) H_x} - Q^{\alpha_m}} = O(\epsilon_1)$) and the subadditivity of error when implementing unitaries.  The fact that each exponential has error $O(\epsilon_1)$ follows from the inequality $\norm{e^{iA} - e^{iB}} \leq \norm{A-B}$. This can be proved by observing that $\norm{e^{iA} - e^{iB}} =  \norm{(e^{iA/n})^n - (e^{iB/n})^n} \leq n\norm{e^{iA/n} - e^{iB/n}} \leq \norm{A-B} + O(1/n)$, where the first inequality uses subadditivity of error and the second inequality follows by Taylor expansion. Since the statement is true for all $n$, the claim follows.

This simulation has continuous-query complexity $T$.  Its error can be made less than $\epsilon$ by choosing $\epsilon_1$ sufficiently small (in particular, it suffices to take some $\epsilon_1 = \Theta(\epsilon/m)$).
\end{proof}

\subsection{The Approximate Segment Lemma}
\label{app:approxsegment}

In this section, we establish the Approximate Segment Lemma (\lem{approxsegment}). This lemma essentially follows from \cite{CGM+09} with minor modification. We start by proving the following Gadget Lemma, which follows from \cite[Section II.B]{CGM+09}.

\GADGET*

\begin{proof}
The input state evolves as follows:
\begin{align}
  |0\>|\psi\>
  &\mapsto \frac{\sqrt{c} |0\> + \sqrt{s} |1\>}{\sqrt{c+s}} |\psi\> \nonumber \\
  &\mapsto \frac{1}{\sqrt{c+s}}(\sqrt{c}|0\>|\psi\> + \sqrt{s}|1\>Q|\psi\>) \nonumber \\
  &\mapsto \frac{1}{c+s}[|0\>(c |\psi\> + i s Q|\psi\>) + \sqrt{cs} |1\> (|\psi\> - i Q|\psi\>)] \nonumber \\
  &= \sqrt{q_\alpha} (|0\>e^{i \pi \alpha/2}Q^\alpha|\psi\> + \sqrt{\sin(\pi\alpha)}|1\> e^{-i \pi/4} Q^{-1/2}|\psi\>).
\end{align}
Thus the output has the stated form.
\end{proof}

We can now collect these gadgets into a segment, which implements a fractional-query algorithm with constant query complexity with amplitude 1/2.

\SEGMENT*

\begin{proof}
We first analyze the subcircuit in the dashed box in \fig{segment}, which is the entire circuit without the first qubit. The first qubit does not interact with the rest of the qubits and is only used at the end of the proof.

This subcircuit is built by composing several fractional-query gadgets (as in \fig{gadget}) with a new control qubit for each gadget but with a common target.  The $m$ gadgets correspond to making the fractional queries $Q^{\alpha_i}$. The first register of a gadget indicates whether it has applied the fractional query successfully, in which case the register is $\ket{0}$, or not, in which case it is $\ket{1}$.  For the $i$th gadget, the output state has amplitude $q_{\alpha_i}$ on the state $\ket{0}$ corresponding to the successful outcome, as shown in \lem{gadget}.

The state of the control qubits on the output is $|0^{m}\>$ only when all the gadgets have successfully applied the fractional query. In this case, the target has been successfully  transformed to $V|\psi\>$.  Thus the dashed subcircuit in \fig{segment} performs the map
\be
|0^{m}\>|\psi\> \mapsto \sqrt{p}|0^{m}\>e^{i\vartheta}V|\psi\> + \sqrt{1-p}|\Phi^\perp\>
\ee
for some $|\Phi^\perp\>$ satisfying $(|0^m\>\<0^m| \otimes \id)|\Phi^\perp\>=0$, where $p = \prod_{i=1}^{m} {q_{\alpha_i}}$ and $\vartheta = -\sum_{i=1}^{m} \pi\alpha_i/2 \bmod 2\pi$.

This is similar to the desired statement, except that we want the amplitude in front of $|0^{m}\>$ to be $1/2$ instead of $\sqrt{p}$. We show that $p>1/4$ and then use the first qubit to decrease its value to exactly $1/4$.

Since $\sum_{i=1}^{m} \alpha_i \leq 1/5$ by assumption, we can lower bound the value of $p$ as follows. Since $\alpha_i \geq 0$ for all $i$, using the inequalities $\sin x \leq x$ (for $x \geq 0$) and $1/(1+x) \geq 1-x$ (for $x \ge -1$) gives
\be
  p = \prod_{i=1}^{m} {q_{\alpha_i}} 
    = \prod_{i=1}^{m} \frac{1}{1+\sin(\pi\alpha_i)} 
    \geq \prod_{i=1}^{m} \frac{1}{1+\pi\alpha_i} 
    \geq \prod_{i=1}^{m} (1-\pi\alpha_i) 
    \geq 1 - \pi \sum_{i=1}^{m} \alpha_i 
    \geq 1 - \frac{\pi}{5} 
    > \frac{1}{4},
\ee
where the third inequality uses the fact that for $x_i \in [0,1]$, $\prod_i (1-x_i) \geq 1 - \sum_i x_i$.

Thus we have $\sqrt{p} > 1/2$. Now let $\fudge$ be any unitary that maps $\ket{0}$ to $\frac{1}{2\sqrt{p}}\ket{0} + (1-\frac{1}{4p})^{1/2}\ket{1}$.  Since $\sqrt{p} > 1/2$, we have $\frac{1}{2\sqrt{p}} < 1$, so a unitary $\fudge$ exists. Then for the full circuit in \fig{segment}, the amplitude corresponding to the state $\ket{0^m}$ is $\sqrt{p} \cdot \frac{1}{2\sqrt{p}} = 1/2$.
\end{proof}

Finally, we show that the map in the previous lemma can be performed to error $\epsilon$ using only  $O\big(\frac{\log(1/\epsilon)}{\log\log(1/\epsilon)}\big)$ queries.

\APPROXSEGMENT*

\begin{proof}
From \lem{segment} we know that the circuit in \fig{segment} performs the claimed map with no error. However, the circuit makes $m$ discrete queries, which can be arbitrarily large. We wish to construct a circuit with error at most $\epsilon$ that makes only $O\big(\frac{\log(1/\epsilon)}{\log\log(1/\epsilon)}\big)$ queries, independent of $m$.

We first analyze the subcircuit in the dotted box in \fig{segment}. The output of this subcircuit is $\ket{\zeta} = \bigotimes_{i=1}^{m} R_{\alpha_i} |0\> = \bigotimes_{i=1}^{m} \frac{1}{\sqrt{c_i+s_i}}(\sqrt{c_i}|0\> + \sqrt{s_i}|1\>)$, where $c_i \defeq \cos(\pi \alpha_i/2)$ and $s_i \defeq \sin(\pi \alpha_i/2)$. We also define $q_i \defeq q_{\alpha_i} =  1/(c_i+s_i)^2 = 1/(1+\sin(\pi\alpha_i))$.  We can write $\ket{\zeta} = \sum_{x \in \{0,1\}^{m}} w_x|x\>$ with $\sum_x |w_x|^2 = 1$.

Now consider the subnormalized state $\ket{\zeta_k} \defeq \sum_{|x| \leq k} w_x|x\>$, where $|x|$ denotes the Hamming weight of $x$ and $k \leq m$ is a positive integer. In the circuit, we approximate the state $|\zeta\>$ with some $|\zeta_k\>$.  Clearly $|\zeta_m\> = |\zeta\>$, and the approximation becomes worse as $k$ decreases.  To achieve a $1-\epsilon^2/2$ approximation, we claim it suffices to take $k = \Omega\big(\frac{\log(1/\epsilon)}{\log\log(1/\epsilon)}\big)$.  Since $1 - \<\zeta|\zeta_k\> = \sum_{|x|>k} |w_x|^2$, we must upper bound $\sum_{|x|>k} |w_x|^2$ in terms of $k$.

Consider $m$ independent random variables $X_i$ with $\Pr(X_i = 0) = \frac{c_i}{c_i+s_i}$ and $\Pr(X_i = 1) = \frac{s_i}{c_i+s_i}$.  The probability that $\sum_i X_i > k$ is $\sum_{|x|>k} |w_x|^2$, since $|w_x|^2$ is the probability of the event $X_i = x_i$ for all $i$.  For such events, the Chernoff bound (see for example \cite[Theorem 4.1]{MR95}) says that for any $\delta>0$,
\be 
\Pr \(\sum_i X_i > (1+\delta)\mu\) < \frac{e^{\delta\mu}}{(1+\delta)^{(1+\delta)\mu}},
\ee 
where $\mu \defeq \sum_i \Pr(X_i = 1) = \sum_i \frac{s_i}{c_i+s_i}$. Since $\alpha_i \geq 0$ and $\sum_i \alpha_i \leq 1/5$, we have $\mu \geq  0$ and $\mu = \sum_i  \frac{s_i}{c_i+s_i} \leq \sum_i s_i = \sum_i \sin(\pi\alpha_i/2) \leq \sum_i \pi\alpha_i/2 \leq \pi/10 \leq 1$, where we used the facts that $\sin x \leq x$ for all $x>0$ and $\sin \theta + \cos \theta \geq 1$ for all $\theta \in [0,\pi/2]$.

Setting $k = (1+\delta)\mu$, we get $\sum_{|x|>k} |w_x|^2 = \Pr (\sum_i X_i > k) < {e^{k-\mu}}/{(1+\delta)^{k}} = {e^{k-\mu}\mu^k}/{k^k} < e^k/k^k$. This is less than $\epsilon^2/2$ when $k = \Omega\big(\frac{\log(1/\epsilon)}{\log\log(1/\epsilon)}\big)$.
For such a value of $k$, the state $\ket{\zeta_k}$ has inner product at least $1-\epsilon^2$/2 with $\ket{\zeta}$. Let $|\tilde{\zeta}\>$ denote the normalized $\ket{\zeta_k}$ for some choice of $k = \Omega\big(\frac{\log(1/\epsilon)}{\log\log(1/\epsilon)}\big)$. The state $|\tilde{\zeta}\>$ also has inner product at least $1-\epsilon^2/2$ with $\ket{\zeta}$. We replace the dotted box in \fig{segment} with $|\tilde{\zeta}\>$, a fixed state that requires no queries to create.

With this modification, the control qubits are in a superposition over states with Hamming weight at most $k$, suggesting that this circuit can be performed with at most $k$ queries. We now show that this is possible.

The control qubits are in a superposition over states $\ket{\x}$ where $\x \in \{0,1\}^{m}$.  The value of $\x_i$ decides whether the $i$th query occurs or not.  The string $\x$ therefore completely determines the product of unitary matrices that is applied to $|\psi\>$ when the control qubits are in the state $|\x\>$.  This product contains at most $k$ query gates, and thus may be written as
\be
  W_{|\x|}(\x) \, Q \, W_{{|\x|}-1}(\x) \cdots Q \, W_1(\x) \, Q \, W_0(\x).
\ee
Note that the $W_i$ operators are functions of $\x$.  We may also write this unitary as
\be
  W_k(\x) \, Q_k(\x) \, W_{k-1}(\x) \cdots Q_2(\x) \, W_1(\x) \, Q_1(\x) \, W_0(\x),
\label{eq:rearranged}
\ee
where for $i \le |\x|$ the $W_i$ operators are as before and for $i > |\x|$, we have $W_i = \id$. Here $Q_i(\x)$ is defined to be $Q$ when $i \leq |\x|$ and $\id$ when $i>|\x|$. We can now construct a circuit that performs the unitary in \eq{rearranged} controlled on the value of $\x$.  This circuit has at most $k$ query gates and performs the same unitary as the circuit in \fig{segment} with $\ket{\zeta}$ replaced with $|\tilde{\zeta}\>$.

Finally, we show that the actual operation performed, denoted $\approxU$, is within error $\epsilon$ of the ideal unitary $U$. The only difference between these operations is that $\approxU$ prepares $|\tilde\zeta\rangle$ rather than $\ket{\zeta}$ in the initial step. Therefore the error between $\approxU$ and $U$ is at most the error between an operation that prepares $|\tilde\zeta\rangle$ and an operation that prepares $\ket{\zeta}$. If we required $U$ to prepare $\ket{\zeta}$ using $\bigotimes_{i=1}^m R_{\alpha_i}$, it would be difficult to design a nearby unitary that prepares $|\tilde\zeta\rangle$. However, the lemma does not specify the action of $U$ on states not of the form $|0^{m+1}\rangle\ket{\psi}$, so we can make any convenient choice of the operation preparing $\ket{\zeta}$ that is close to the operation preparing $|\tilde\zeta\rangle$.

Let $R:=\bigotimes_{i=1}^m R_{\alpha_i}$ and denote the unitary that prepares $|\tilde\zeta\rangle$ by $\tilde R$. In the computational basis, $R$ has first column $\zeta$ and $\tilde R$ has first column $\tilde\zeta$. We claim there is a unitary $R'$ that is within $\epsilon$ of $\tilde R$ but that has the same first column as $R$.

To see this, let $\theta$ satisfy $\<\tilde\zeta|\zeta\>=\cos\theta$. Consider the 2-dimensional subspace spanned by $|\zeta\>$ and $|\tilde\zeta\>$, and let $E$ be the unitary that rotates by angle $\theta$ in this subspace, but acts as the identity outside the subspace. In particular, $E|\tilde\zeta\>=|\zeta\>$. Taking $R' := E\tilde R$, we see that $R'$ has the first column $\zeta$ as required. The error is $\|R'-\tilde R\|=\|E\tilde R-\tilde R\|=\|E-\id\|=\sqrt{2-2\cos\theta}$.

Since $\<\tilde\zeta|\zeta\> \ge 1-\epsilon^2/2$, we find $\|R'-\tilde R\|\le\epsilon$. Because the remainder of the circuit is identical, the overall error between $\approxU$ and $U$ is at most $\epsilon$ as claimed.
\end{proof}

\end{document}